\documentclass[twocolumn]{autart}

\usepackage{graphicx} 
\usepackage{xcolor}
\usepackage{mathtools}
\usepackage{amsfonts,amsmath,amssymb}
\usepackage{dsfont}
\usepackage[caption=false]{subfig}
\usepackage{mathrsfs}

\usepackage{soul}

\usepackage{pgfplots}
\usetikzlibrary{plotmarks}
\pgfplotsset{compat=1.8}

\newcommand{\iu}{\mathrm{i}\mkern1mu}
\newcommand{\one}{\mathds{1}}
\newcommand{\ket}[1]{\lvert #1 \rangle}
\newcommand{\kket}[1]{\lvert #1 \rangle\rangle}
\newcommand{\bra}[1]{\langle #1 \rvert}
\newcommand{\bbra}[1]{\langle\langle #1 \rvert}
\newcommand{\ketbra}[2]{\lvert #1 \rangle \langle #2 \rvert}

\newcommand{\tr}{\mathrm{tr}}

\newcommand{\qedsymbol}{\leavevmode
  \hbox to.77778em{%
  \hfil\vrule
  \vbox to.675em{\hrule width.6em\vfil\hrule}%
  \vrule\hfil}}
\newcommand{\mathqed}{\quad\hbox{\qedsymbol}}
\DeclareRobustCommand{\qed}{%
  \ifmmode \mathqed
  \else
    \leavevmode\unskip\penalty9999 \hbox{}\nobreak\hfill
    \quad\hbox{\qedsymbol}%
  \fi
}

\usepackage{url}

\definecolor{color1}{rgb}{0.0, 0.6056031704619725, 0.9786801190138923}
\definecolor{color2}{rgb}{0.8888735440600661, 0.435649148506399, 0.2781230452972764}
\definecolor{color3}{rgb}{0.24222393333911896, 0.6432750821113586, 0.304448664188385}
\definecolor{color4}{rgb}{0.7644400000572205, 0.4441118538379669, 0.8242975473403931}
\definecolor{color5}{rgb}{0.6755439043045044, 0.5556622743606567, 0.09423444420099258}

\newcommand{\Ac}{\mathcal{A}}
\newcommand{\Bc}{\mathcal{B}}

\newcommand{\Hc}{\mathcal{H}}

\newcommand{\Kc}{\mathcal{K}}
\newcommand{\Lc}{\mathcal{L}}
\newcommand{\Mc}{\mathcal{M}}

\newcommand{\Oc}{\mathcal{O}}
\newcommand{\Pc}{\mathcal{P}}
\newcommand{\Qc}{\mathcal{Q}}

\newcommand{\Sc}{\mathcal{S}}
\newcommand{\Tc}{\mathcal{T}}
\newcommand{\Uc}{\mathcal{U}}
\newcommand{\Vc}{\mathcal{V}}
\newcommand{\Wc}{\mathcal{W}}
\newcommand{\Xc}{\mathcal{X}}

\begin{document}

\begin{frontmatter}

\title{Quantum model reduction based on Oja's flow}

\author[unipd]{Miguel Casanova}\ead{casanovame@dei.unipd.it}
\author[kyodai]{Kentaro Ohki}\ead{ohki@tokai.ac.jp}
\author[unipd,qtech]{Francesco Ticozzi}\ead{ticozzi@dei.unipd.it}

\address[unipd]{Department of Information Engineering, University of Padova, Italy}
\address[kyodai]{Department of Applied Computer Engineering, Tokai University, Japan}
\address[qtech]{QTech Center, University of Padova, Italy}

\begin{keyword}
Model reduction; Slow dynamics; Adiabatic elimination; Quantum dynamical semigroups; Oja's flow.
\end{keyword}

\begin{abstract}
We propose a novel approach to numerically derive approximate reduced dynamical models for Markovian quantum open systems without perturbative iterations, projecting the evolution to the subspace associated to their slowest degrees of freedom. The two algorithms we develop are based on Oja’s continuous-time principal component flow: the first returns the optimal reduction to the slowest decaying operator-subspace, and is extended to time-dependent dynamics, while the second one is designed to reduce the dynamics on a subspace of the system's Hilbert space, and thus preserve conditional complete positivity. The methods represent a non-perturbative alternative to well-established Adiabatic Elimination (AE) methods, and the second can be used to find noise-protected subspace codes for quantum information processing. Both are tested on a paradigmatic central spin model.
\end{abstract}

\end{frontmatter}

\section{Introduction}

The search for viable ways to effectively simulate complex quantum dynamics has arguably been the first motivation behind the current push for the development of quantum technologies \cite{Feynman82}. Such a simulation task becomes unfeasible as soon as the physical system of interest is infinite-dimensional (e.g., resonator cavities \cite{ScullyZubairy1997}), or is a large many-body system, whose dimension increases exponentially on the amount of subsystems (e.g., spin chains \cite{Lopez2021}). These scenarios, as well as the impossibility of accurately modeling the environmental degrees of freedom, have led to the development of a multitude of approaches to obtain reduced, treatable descriptions, both theoretically - mainly with a wide variety of quantum statistical master equations - and numerically \cite{breuerTheoryOpenQuantum2007,rivasOpenQuantumSystems2012,Orus2014,VerstraeteMurgCirac2008,Alicki1987-ui,algebraic-0,algebraic-1,Nurdin2014,Tamascelli2017,TanimuraKubo1989}.

Focusing on dynamics associated to a master equation in Lindblad form, a comprehensive framework for exact model reduction has
been developed leveraging Krylov  subspaces and algebraic
methods in \cite{algebraic-0,algebraic-1}.
However, the resulting conditions for the
existence of an exact reduced model may be too strong,
and forcing the reduction to be exact is too limiting whenever the original model is already approximate,
e.g., constructed from experimental data.
Approximate model reduction is often possible if the dynamics
of interest exhibit a time-scale separation, that is, if there is a significant gap in the spectrum of the dynamical generator. One could then simulate only the dynamics restricted to the slow manifold (e.g., the center manifold or a larger invariant subspace containing it), excluding faster decaying degrees of freedom. This is the goal, for instance, of Adiabatic Elimination (AE) methods \cite{Hope2000,azouitGenericAdiabaticElimination2017,finkelstein-shapiroAdiabaticEliminationSubspace2020}, which are {\em analytical} but typically require iterated approximations, exploiting a representation of the dynamics of interest as a nominal generator, which is assumed to have a known center manifold, plus a small perturbation.

In this work, we take a different approach: we develop two {\em numerical algorithms} to perform model reduction of a large yet finite-dimensional quantum system, resulting in a model that describes only the slowest degrees of freedom. The first is a direct application of Oja’s continuous-time principal component flow \cite{tsuzukiLowrankApproximatedKalmanBucy2024} to the matrix representation of the generator of a quantum dynamical semigroup (QDS): the reduced model we obtain is a linear one, defined on the subspace of the slowest decaying {\em operator} of the original system. The Oja flow, which we recall in detail later, is a matrix-valued dynamical system, defined in a suitable manifold of fixed dimensions, which asymptotically converges to a basis for the principal components of the original model to which it is applied. This algorithm provides a reduction on the slow subspace, making it equivalent to an infinite-order adiabatic elimination. The method is then extended to obtain reductions of time-dependent generators. The second algorithm is obtained as a structured modification of the Oja-flow method, in which the particular matrix structure we impose guarantees that the reduced model remains physically admissible, that is, the generator remains conditional complete positive (CCP) \cite{wolfQuantumChannelsOperations2012}. It does so by forcing the reduction to select the most relevant, slowest decaying subspace of the underlying {\em vector} space. This second method is shown to be equivalent to finding the best (according to a particular cost function) approximate unitarily evolving subspace of a given dimension within the original dynamics.

 With respect to existing AE methods, our method is numerical in nature and does not assume any specific perturbation-dependent form of the generator, and thus no iterative approximation procedure is required. The approach is thus systematic and has a quadratic advantage, in terms of memory resources needed for the computations, with respect to direct diagonalization/Jordan decomposition.
 Furthermore, our method is well adapted to sparse matrix
 representations of the dynamical generator, significantly
 reducing the computational requirements of the algorithm,
 whereas direct diagonalization may require the
 use of dense matrices
 whenever any of the eigenvalues is degenerate.
 In addition, the second method is able to retain the physical character of the evolution, as the resulting reduced semigroup is guaranteed to be completely positive and trace preserving (CPTP). This is a particularly desirable feature, as it allows for direct interpretation and potential implementation on an open quantum simulator. In contrast, it is known that it is not always possible to perform positivity-preserving AE \cite{tokiedaCompletePositivityViolation2024}.

The paper is organized as follows. Section \ref{sec:model} presents the class of models of interest, quantum open Markovian dynamics and their generators. Section \ref{sec:motiv} presents the motivating applications: model reduction to the slow manifold and the quest for approximate decoherence-free subspaces. The general Oja flow framework is recalled in Section \ref{sec:oja}, while our first method is developed in Sections \ref{sec:ojaL} and \ref{sec:timedep}, with the latter dedicated to the time-dependent case. The second method, where we can ensure the physical CPTP character of the dynamics, is constructed in Sections \ref{sec:lindblad_conds} and \ref{sec:posflow}, including a discussion and comparison with the adiabatic method.
To test the validity of our method, in Section \ref{sec:simulations} we consider a central spin model where a $1/2$-spin is coupled to a bath of dissipative $1/2$-spins by $XX$ and $ZZ$ interactions.

\section{Quantum Markovian dynamics}\label{sec:model}

We here provide a brief introduction to the quantum models of interest in this work.
The mathematical description of a quantum system is built on a Hilbert space $\Hc,$ which we here consider to be potentially large but {\em finite-dimensional}.
The elements of these Hilbert spaces are represented as ``kets'' $\ket{\psi} \in \Hc$,
and the elements of the dual (linear functionals) as ``bras'' $\bra{\psi} \in \Hc^*$. For finite-dimensional
systems, kets are isomorphic to column vectors in $\mathbb{C}^n$, and bras are to row vectors
equal to the conjugate-transpose of the ket, $\bra{\psi} = \ket{\psi}^\dagger$. 

The state of an open quantum system is then described by a density matrix $\rho \in \mathfrak{D}(\Hc) \subset \Bc(\Hc)$,
which is Hermitian and positive-semidefinite,
and has trace equal to $1$. Here, $\mathfrak{D}(\Hc)$ denotes the (convex) set of all density matrices, and $\Bc(\Hc)$ the space
of bounded linear operators acting on $\Hc$. 
The density matrices are the noncommutative analogous to probability distributions.
Measurable quantities are described by Hermitian operators $O \in \mathfrak{h}(\Hc) \subset \Bc(\Hc)$, called observables.
The expected value of an observable can then be computed as $\mathbb{E}_{\rho}[O] = \tr(O \rho)$.
Here, $\mathfrak{h}(\Hc)$ stands for the set of all Hermitian (bounded linear) operators acting on $\Hc$.
Observables are analogous to random variables.

General time-homogeneous Markovian open quantum dynamics are associated to quantum dynamical semigroups (QDS) $\{\Phi_t\}_t$, a one-parameter family of continuous, completely positive (CP),
and trace preserving (TP) maps\footnote{A linear map $\Phi: \Bc(\Hc) \to \Bc(\Hc)$ is said to be CP if
$\Phi \otimes \one_n: \Bc(\Hc) \otimes \mathbb{C}^{n \times n} \to \Bc(\Hc) \otimes \mathbb{C}^{n \times n}$ is positive for all $n \geq 1$,
and it is TP if $\tr(\Phi(\rho)) = \tr(\rho)$ \cite{nielsenQuantumComputationQuantum2012}.} acting on density matrices such that $\Phi_0 = \one$ and $\Phi_t \Phi_s = \Phi_{t+s}$.
The CPTP property ensures that probabilities remain positive and normalized, even if the system is initially entangled with its environment \cite{nielsenQuantumComputationQuantum2012}.

The generator of a QDS $\Lc$ is a superoperator such that $\Phi_t = e^{\Lc t}$, and is commonly called
Gorini-Kossakowski-Sudarshan-Lindblad (QKSL) generator \cite{alickiQuantumDynamicalSemigroups2007}, or simply {\em Lindbladian}. Explicitly, the associated differential equation can be written as
\begin{equation}
\label{eq:lindblad}
\dot \rho = \Lc(\rho) = -\iu [H_{\rm nh}, \rho] + \sum_m L_m \rho L_m^\dagger,
\end{equation}
where $[A, B] = AB - BA$ is called the commutator of $A$ and $B$,
and $H_{\rm nh} = H - \frac{\iu}{2} \sum_m L_m^\dagger L_m$.
The operator $H$ is the Hamiltonian of the system, which
describes the internal energy of the system. The operators $L_m$ are called dissipation or noise operators,
 and describe the interaction with a Markovian environment.

It is possible to prove that a map $\Lc$ is a Lindbladian if and only if it satisfies the following three conditions \cite{evansQuantumSymmetriesOperator1998}:
\begin{enumerate}
\item {\em Hermiticity:} $\Lc(\rho^\dagger) = \Lc(\rho)^\dagger$.
\item {\em Trace preservation:} $\Lc^\dagger(\one) = 0$.
\item {\em Conditional complete positivity:} The $\Bc(\Hc) \times \Bc(\Hc) \to \Bc(\Hc)$ kernel: \[(\sigma, \tau) \mapsto \Lc(\sigma^\dagger \rho \tau) + \sigma^\dagger \Lc(\rho) \tau - \sigma^\dagger \Lc(\rho \tau) - \Lc(\sigma^\dagger \rho) \tau\] is positive definite, where $\rho = \rho^\dagger \geq 0$.
\end{enumerate}
Conditional complete positivity (CCP) is the differential equivalent to CP. In fact, one can prove that $\Lc$
is CCP if and only if $e^{\Lc t}$ is CP for all $t \geq 0$.

Being a semigroup of contractions \cite{alickiQuantumDynamicalSemigroups2007}, the spectrum of any QDS generators is composed of either real or complex-conjugate pairs
of eigenvalues that lie on the left half of the complex plane, including the imaginary axis, i.e.
$\mathrm{Re}(\lambda_i) \leq 0$. In fact, if all the eigenvalues lie on the imaginary axis, then we
have that dynamics are purely Hamiltonian. Otherwise, we have $L_m \neq 0$.
We also know that there is at least one eigenvalue equal to zero
and that the zero eigenvalues are non-defective, i.e., their arithmetic multiplicity does not exceed
the geometric multiplicity.

When control actions are present, they are typically modeled as time-dependent terms in the Hamiltonian $H(t) = H_0 + \sum_i u_i(t) H_i$,
where $u_i(t)$ are the input signals. The resulting Lindbladian $\Lc_t$ is no longer the generator of
a QDS $\{\Phi_t\}_t$ as described above, but of a two-parameter semigroup $\{\Phi_{t_f, t_0}\}_{t_f, t_0}$,
where $\Phi_{t_f, t_0} = \Tc e^{\int_{t_0}^{t_f} \Lc_t dt}$, and $\Tc e$ is the time-ordered exponential.
Notice that despite the time dependence, it still is the generator of Markovian dynamics, since there is
no dependence on the past history of the state of the system.

Since $\Lc$ is a linear super-operator, we can find a matrix representation by column-wise vectorization,
i.e., stacking vertically the columns of $\rho$ into a single column vector ${\rm vec}(\rho)$, of the GKLS form.
\begin{equation}
\label{eq:veclindblad}
\kket{\dot \rho} = \hat \Lc \kket{\rho},
\end{equation}
where $\kket{\rho} = {\rm vec}(\rho) \in \Hc \otimes \Hc$, and
$\hat \Lc = -\iu (\one \otimes H_{\rm nh} - H_{\rm nh}^* \otimes \one) + \sum L_m^* \otimes L_m \in \Bc(\Hc \otimes \Hc)$,
which is obtained by using the property ${\rm vec}(A X B) = (B^\top \otimes A) {\rm vec}(X)$.
We use the double ket $\kket{\rho}$ and the hat notations $\hat \Lc$ in order to make the matrix representation explicit when used.
This matrix representation is  commonly used in the quantum information and control
literature and, in contrast to e.g., coherence vector representations,
allows for a direct characterization of CCP that we exploit in Sections
\ref{sec:lindblad_conds} and \ref{sec:posflow}.

\section{Slow dynamics and motivating applications}\label{sec:motiv}
\subsection{Separation of time scales for linear dynamics}
Consider a linear, time-invariant dynamical system with a partitioned state vector $x = [x_s^\top x_f^\top]^\top$, where $x_f$ converges fast to a steady state $x_{f, \infty}$,
while $x_s$ has a much slower decay.
Then, the dynamics of $x_s$ may be approximated by means of
separation of time-scales \cite{kristiansenReviewMultipleTimeScaleDynamics2023}. The state equation of the system is the following,
\begin{equation}
\frac{d}{dt} \left[\begin{smallmatrix} x_s \\ x_f \end{smallmatrix}\right] =
\left[\begin{smallmatrix} A_{11} & A_{12} \\ A_{21} & A_{22} \end{smallmatrix}\right]
\left[\begin{smallmatrix} x_s \\ x_f \end{smallmatrix}\right].
\end{equation}

Then, one can impose the following in order to study the dynamics after $x_f$ has reached its steady state $x_{f, \infty}$
\begin{equation}
A_{21} x_s + A_{22} x_{f, \infty} = 0.
\end{equation}
This gives us the following expression for $x_{f, \infty}$
\begin{equation}
x_{f, \infty} = - A_{22}^{-1} A_{21} x_s,
\end{equation}
which can be used to find a reduced model describing the dynamics of $x_s$ only, that is, the {\em slow dynamics}:
\begin{equation}
\label{eq:timescalesep}
\dot x_s = A_{11} x_s - A_{12} A_{22}^{-1} A_{21} x_s.
\end{equation}

\subsection{Adiabatic elimination}

In dissipative open quantum systems, when there is a part of the system that decays much faster than the rest,
it is possible to simplify the analysis of the dynamics by performing \textit{adiabatic elimination} of the fast variables. The reduced dynamics will then become asymptotically exact, after an exponentially decaying error in the transient. The approach can be understood as a generalization of the \textit{separation of time-scales} technique used in classical systems, which we reviewed above, to the case where
the slow-fast subsystem decomposition is not exactly computable (e.g. the system size is too large to perform a Jordan decomposition), but a series expansion can be performed to approximate it.
There are different methods by which it can be done, but they all lead to
similar results, at least in the first few orders of the approximation.

Most AE methods, rely on performing a perturbative series
expansion in terms of a small parameter $\varepsilon$ that weighs either the fast decaying part of the system
\cite{azouitGenericAdiabaticElimination2017,rivaExplicitFormulasAdiabatic2024}
or the non-diagonal part of the QDS generator \cite{kesslerGeneralizedSchriefferWolffFormalism2012}.
More precisely, the Lindbladian $\Lc$ is assumed to be of the form $\Lc = \Lc_0 + \varepsilon \Lc_\varepsilon$,
and the knowledge about $\Lc_0$ is used to define an approximate projector $\Pc$ onto the invariant slow manifold
of the system, which is then further refined by perturbative expansion on $\varepsilon$. Alternatively,
the approximate projector $\Pc$ can be assumed to be known, while assuming no particular structure of $\Lc$
(besides the existence of a spectral gap, which is always necessary in order to have distinct slow and fast dynamics).
Then, a series expansion is used to refine the approximation of the dynamics of $\Pc \rho$
\cite{finkelstein-shapiroAdiabaticEliminationSubspace2020,saidehProjectionbasedAdiabaticElimination2020}.
All of these approaches lead to
equivalent expressions in the first few orders of the approximation (even if they may differ at higher orders), namely:
\[\Lc_{\rm eff} = \mathcal{P L P} - \mathcal{P L Q} (\mathcal{Q L Q})^{-1} \mathcal{Q L P},\] where $\Qc = \one - \Pc.$ A derivation for weakly interacting quantum systems is provided in the Appendix \ref{sec:example}.
Notice that this result is equivalent to the one obtained with the separation of time-scales method, by posing
$x_s = \Pc \rho$, $A_{11} = \mathcal{P L P}$, $A_{12} = \mathcal{P L Q}$, $A_{21} = \mathcal{Q L P}$ and $A_{22} = \mathcal{Q L Q}$.

However, notice that the effective generators $\Lc_{\rm eff}$ obtained with AE
are not QDS generators in general, especially at higher orders of the approximation. Indeed, it is a well-known
fact in the AE literature that conditional complete positivity (CCP), in particular,
cannot always be guaranteed for every expansion order \cite{tokiedaCompletePositivityViolation2024}.
But the conditions under which this is the case are not completely clear.
In Section \ref{sec:lindblad_conds}, we provide 
sufficient conditions for the preservation of CCP in the reduced model, which are equivalent to a (non-perturbative) projection of the dynamics onto a Hilbert subspace.
Then in Section \ref{sec:posflow}, we exploit this condition to derive a (suboptimal) reduced, CCP preserving flow.
Notice also that if the exact projector $\Pc$ on the slow manifold is known,
we have $A_{21} = \mathcal{Q L P} = 0$, due to the invariance of this manifold.
As such, the effective QDS generator becomes simply $\Lc_{\rm eff} = \Pc \Lc \Pc$.
In the following, we show how Oja's flow can be used to compute $\Pc$ numerically
without the need of costly operations such as diagonalization or inversion of large matrices,
which is the main obstacle in finding purely numerical AE algorithms for general (gapped)
quantum dynamical systems.

\subsection{Noise protected codes in quantum information}
In quantum information processing, the effect of uncontrolled environmental interactions are typically detrimental to the computational task, as they degrade the quantum coherences (state superpositions) that offer quantum algorithms their advantage. A wide range of error protection and correction techniques has been devised to protect suitably encoded quantum information \cite{lidarQuantumErrorCorrection2013}, including the use of noiseless codes: by exploiting symmetries in the system, one can find the degrees of freedom that are the least affected by the environment. The simplest class of such codes corresponds to noise-protected subspaces \cite{lidarDecoherenceFreeSubspacesQuantum1998}:
Consider a Hilbert space $\Hc$ that may be decomposed as $\Hc \simeq \Hc_\mathsf{q} \oplus \Hc_\mathsf{r}$.
If any state $\rho$ that is initialized with support only on $\Hc_\mathsf{q}$ evolves unitarily,
then $\Hc_\mathsf{q}$ is said to be a {\em Decoherence-Free Subspace (DFS).}

For this to happen, $\Hc_\mathsf{q}$ must support an invariant subspace for the generator $\Lc,$ and the reduction $\Lc|_{\Hc_\mathsf{q}}$ must be Hamiltonian, that is, have a purely imaginary spectrum.
One can then, in principle, identify DFS's for a given $\Lc$ by finding its center manifold (subspace). However, this is a hard computational task for large systems. Furthermore, exact DFS are not found in any realistic scenarios, and in particular when a model is reconstructed from noisy data. In these situations, one is left searching for a subspace that exhibits the slowest possible decoherence, maximizing the available computational time. That is, we aim to find {\em approximate DFS} \cite{casanovaQuantumWallStates2026,wangNumericalMethodFinding2013}. Our Oja flow approach can be used to this aim, bypassing the need for computationally heavier Jordan decomposition of the generator, as we shall argue in the next Sections.

\section{Oja's flow for classical linear systems}
\label{sec:oja}

Consider a simply stable LTI dynamical system, exhibiting distinct fast and slow dynamics, with state equation
\begin{equation}
\dot x = A x,
\end{equation}
where $A \in \mathbb{C}^{m \times m}$. Since the system is simply stable, the eigenvalues of $A$ have
non-positive real part.
Let $A \psi_i = \lambda_i \psi_i$, such that
${\rm Re}(\lambda_i) \geq {\rm Re}(\lambda_{j})\ \forall i > j$.
We say that there is a {\em spectral gap} $\Delta_\ell$ after the $r$-th ordered
eigenvalue if 
$\Delta_\ell = \mathrm{Re}(\lambda_\ell) - \mathrm{Re}(\lambda_{\ell+1}) > 0$.
The degrees of freedom associated to eigenvalues $\lambda_i$
with $i > \ell$ are called fast,
whereas those with $i \leq \ell$ are called slow.
We remark that a design parameter of the proposed Oja's flow
based model reduction is precisely
the choice of $\ell$ such that $\Delta_\ell > 0$. In the remainder
of this work we assume such a gap to be present.

Oja's principal component flow \cite{ojaSimplifiedNeuronModel1982} is an iterative technique that uses a dynamical equation to find a lower-dimensional
representation of the matrix $A$.
Define the {\em Oja dynamics} as:
\begin{equation}
\label{eq:ojaflow}
\epsilon \dot V = (\one - V V^\dagger) A V,
\end{equation}
where $V \in \mathbb{C}^{n \times r}$ and where $\epsilon > 0$ is a rate-control parameter.
This parameter can be exploited  to improve the convergence of
the numerical integration schemes when $\Delta_\ell$ is small.
Choosing $\epsilon$ large enough guarantees the convergence of
Oja's flow even for the simplest Euler method. However,
one must respect the constraint that
$\Delta_\ell / \epsilon$ remains distinguishable from a numerical zero.

By simulating Oja's flow until convergence to a steady state $V_\infty \in \mathbb{C}^{n \times r}$,
we are able to obtain a lower-dimensional representation $V_\infty^\dagger A V_\infty \in \mathbb{C}^{\ell \times \ell}$
of the direct sum of the eigenspaces of $A$, associated to the eigenvalues with the greatest real parts, without resorting to a Jordan decomposition.
Classical results about Oja's flow are restricted to the case where $A$ is real and symmetric,
i.e. $ A = A^\top \in \mathbb{R}^{m \times m}$, and $V$ belongs to the real-valued Stiefel manifold ${\rm St}(\mathbb{R}^m, \ell)$,
that is the set of $V \in \mathbb{R}^{m \times \ell}$ such that $V^\top V = \one_\ell$ \cite{wei-yongyanGlobalAnalysisOjas1994}.
For such cases, Oja's flow can be derived as the Riemannian gradient ascent flow $\dot V = {\rm grad} J(V)$
on the Stiefel manifold of the function
\begin{equation}
J(V) = \tr(A V V^\top).
\end{equation}

Since $V$ is a point of the Stiefel manifold, $V V^\top$ is an orthogonal projector, and
$J(V) = \tr(V V^\top A V V^\top)$. Furthermore, since $A$ is symmetric, its eigenvalues are
all real, and its left and right eigenvectors are identical. Then it is easy to see that
$J(V)$ is maximized by choosing $V$ such that $V V^\top$ is a projector onto the eigenspaces
with the largest eigenvalues, i.e., if the columns of $V$ are linear combinations of the eigenvectors
of $A$ with the largest eigenvalues. That is why Oja's flow allows us to extract the principal components
of $A$.

If $A$ is not symmetric, the Riemannian gradient of $J(V)$
no longer
corresponds to Oja's flow. This is due to the fact that $\partial_V \tr(A V V^\top) = (A + A^\top) V$,
introducing $A^\top$ in the gradient. In other words, such a gradient flow would only allow us to find
the principal components of the symmetric part of $A$, i.e. $A_{\rm sym} = A + A^\top$. However,
if we allow $V$ to be complex-valued, i.e. $V \in {\rm St}(\mathbb{C}^{m}, \ell) \subset \mathbb{C}^{m \times \ell}$,
then we recover Oja's flow \eqref{eq:ojaflow} as the Riemannian gradient ascent flow of
\begin{equation}
J(V) = \tr(A V V^\dagger)
\end{equation}
for general square matrices $A \in \mathbb{R}^{m \times m}$.

While $J(V)$ for complex $V$ is not a real-valued functional, not even for real-valued $A$,
(and therefore we cannot talk about ``maximizing'' it),
it has recently been proven in \cite{tsuzukiGlobalConvergenceOjas2025} that the set of stable fixed points
of the flow \eqref{eq:ojaflow} is the following,
\begin{equation}
\begin{split}
V_\infty \in
\mathfrak{V} \coloneq \{&[\psi_1, \ldots, \psi_r] [K_\ell^\top, 0]^\top \in {\rm St}(\mathbb{C}^m,\ell), \\
& K_\ell \in \mathbb{C}^{\ell \times \ell}\},
\end{split}
\end{equation}
where, as before, $\psi_j$ are the ordered eigenvectors of $A$ and $\mathrm{St}(\mathbb{C}^m,\ell) \subset \mathbb{C}^{m \times \ell}$ is the Stiefel manifold. More precisely,
Oja's flow with initial point taken from almost everywhere (i.e. with the exception of a set of measure zero)
in the Stiefel manifold converges exponentially to the set $\mathfrak{V}$. The convergence rate is governed by
$(\mathrm{Re}(\lambda_{\ell+1}) - \mathrm{Re}(\lambda_{\ell}) + \delta) / \epsilon$, where $\delta > 0$ is an arbitrarily
small constant.
Furthermore, it is easy to see that the set $\mathfrak{V}$ is invariant under unitary transformations $W \in {\rm U}(\ell)$,
i.e. if $V_\infty \in \mathfrak{V}$, then for any $W \in \mathbb{C}^{\ell \times \ell}$ such that $W^\dagger W = \one$
we have that $V W \in \mathfrak{V}$.

Since $V_\infty \in \mathrm{St}(\mathbb{C}^m, \ell)$, we know that $V_\infty V_\infty^\dagger \in \mathbb{C}^{m \times m}$ is a projector.
In fact, it is the projector $\Pc$ onto the slow manifold that was described in the previous section. As such,
it satisfies $V_\infty V_\infty^\dagger \psi_i = \psi_i \ \forall i \leq \ell$. Furthermore, we also have that the spectrum
of the reduced model $V_\infty^\dagger A V_\infty$ coincides with that of the slow part of $A$, i.e.
\begin{equation}
\lambda_i(V_\infty^\dagger A V_\infty) = \lambda_i(A),\ i = 1, \ldots, \ell.
\end{equation}
If the Hermitian part of $A$, i.e. $\frac{1}{2} (A + A^\dagger)$, is positive-definite,
then integrating \eqref{eq:ojaflow} by forward Euler's method $V_{k+1} = V_k + \Delta t \dot V_k$
is enough to find $V_\infty \in \mathfrak{V}$. However, we are interested in the cases when $A$ is negative-semidefinite, instead.
Whenever $A$ is not positive-definite, one could either substitute $A$ by $A + a \one_m$
with large enough $a$, such that $A + a I > 0$, or one could apply retractions,
which are displacements maps ${\rm Retr}_V(X)$ over a manifold $\Mc$ along direction $X \in T_V \Mc$
(such as the cost-efficient QR\footnote{The QR decomposition of a square matrix $X$ is a product of two matrices $Q R$,
such that $Q^\dagger Q = \one$ and $R$ is upper-triangular. ${\rm qr}_Q(\cdot)$ stands for the $Q$ part of the decomposition.} retraction \cite{linProjectionRobustWasserstein2020},
${\rm Retr}_{V}^{\rm qr}(X) = {\rm qr}_Q(V + X) \in {\rm St(\mathbb{C}^m, \ell)}$)
to maintain the trajectory in the Stiefel manifold. The retracted forward Euler iteration is given by
$V_{k+1} = \mathrm{Retr}_{V_k}(\Delta t \dot V_k)$.
Either of the two strategies preserves the numerical stability of the flow and the attractiveness
of the set $V_\infty$.

\section{Reducing Lindbladians via Oja Flow}\label{sec:ojaL}

Having introduced the properties of Oja's flow, we now proceed to show how it can be applied to quantum systems.
In the following we consider a quantum system of dimension $n$, so that its density operator vectorized dimension is $n^2$. We aim to obtain a reduced model of dimension $\ell$ and, while not strictly required by the method, we will focus on the case $\ell=r^2$. This will allow us to carry out the analysis done in Section \ref{sec:lindblad_conds} and a
direct comparison to the CP-preserving method introduced in Section \ref{sec:posflow}. Therefore, from now on, we actually apply Oja's flow to $n^2 \times n^2$ matrices, and
we obtain $\hat \Vc_\infty \in \mathbb{C}^{n^2 \times r^2}$, where the calligraphic font also points to the
implicit squaring of the dimension clearer.
Additionally, since from this point we transition from general linear
systems to the context of quantum systems, in the following we set $A = \hat \Lc \in \mathbb{C}^{n^2 \times n^2}$,
i.e., from now on, the dynamical matrix $A$ corresponds to the matrix representation of the QDS generator $\Lc$.

Being a complex-valued matrix, a priori we do not have convergence guarantees for the Lindbladian $\hat \Lc$.
However, we next prove a result that gives us a way to unitarily transform it into a real matrix.
Define \[\hat \Sc = \sum_{ij} \ket{i}\ket{j}\bra{j}\bra{i}\] the {\bf ${\rm SWAP}$ gate} between the two copies of $\Hc,$ such that
$\hat \Sc \ket{\psi}\otimes\ket{\phi}=\ket{\phi}\otimes\ket{\psi}$ for any pair of vectors.

\begin{prop}
Let $\hat \Lc \in \Bc(\Hc \otimes \Hc)$ be the matrix representation of a Lindbladian in vectorized space. Then,
\[\hat \Ac = \hat \Sc^{1/2} \hat \Lc \hat \Sc^{{1/2}\dagger}\] is a real-valued matrix.
\end{prop}
\begin{pf}
Recall from \eqref{eq:veclindblad} that $\hat \Lc = -\iu (\one \otimes H_{\rm nh} - H_{\rm nh}^* \otimes \one) + \sum_m L_m^* \otimes L_m$.
The intuition behind this proof comes from noticing that if
we swap the operators acting on each partition,
we obtain the complex conjugate $\hat \Lc^*$, i.e.
\begin{equation*}
\label{eq:swapconj}
\begin{split}
\hat \Sc \hat \Lc \hat \Sc^{\dagger} =& -\iu (H_{\rm nh} \otimes \one - \one \otimes H_{\rm nh}^*) + \sum_m L_m \otimes L_m^*
= \hat \Lc^*.
\end{split}
\end{equation*}
Notice also that $\hat \Sc$ is not only unitary, but also (transpose) symmetric, i.e., $\hat \Sc = \hat \Sc^\top$.
By consequence, there exists a symmetric (and unitary) square root of
$\hat \Sc = \hat \Sc^{1/2} \hat \Sc^{1/2}$ \cite[Corollary 2.5.20]{horn2012matrix}.
Then, we can apply a unitary conjugation with $\hat \Sc^{{1/2}^\dagger}$ to
both sides of \eqref{eq:swapconj}, obtaining
\begin{equation*}
\hat \Sc^{{1/2}^\dagger} \hat \Sc \hat \Lc \hat \Sc^{\dagger} \hat \Sc^{1/2}
= \hat \Sc^{{1/2}^\dagger} \hat \Lc^* \hat \Sc^{1/2},
\end{equation*}
where $\hat \Sc^{{1/2}^\dagger} \hat \Sc = \hat \Sc^{1/2}$ and
$\hat \Sc^{{1/2}^\dagger} = \hat \Sc^{{1/2}^*}$.
Therefore, we have that both sides of the equation are the complex conjugate of the other,
i.e., that $\hat \Sc^{1/2} \hat \Lc \hat \Sc^{{1/2}^\dagger}$ is equal to its conjugate, thus proving that it is real-valued.
\qed
\end{pf}

Then the convergence of Oja's flow for Lindbladian matrices (and any complex matrix that is unitarily similar to a real matrix) is guaranteed by the following corollary.

\begin{cor}
The flow described by $\dot \Vc = (\one - \Vc \Vc^\dagger) \hat \Lc \Vc$, where $\Vc \in {\rm St}(\mathbb{C}^{n^2}, r^2)$ is the pushforward of
Oja's flow $\dot \Vc^\prime = (\one - \Vc^\prime \Vc^{\prime^\dagger}) \Ac \Vc^\prime)$,
by $\hat \Sc^{{1/2}^\dagger}$, that is  $\Vc = \hat \Sc^{{1/2}^\dagger} \Vc^\prime$,
where $\Ac = \hat \Sc^{1/2} \hat \Lc \hat \Sc^{{1/2}^\dagger} \in \mathbb{R}^{n, n}$.
\end{cor}
\begin{pf}
We start from
\begin{equation*}
\begin{split}
\frac{d}{dt} \hat \Vc =& (\one - \hat \Vc \hat \Vc^\dagger) \hat \Lc \hat \Vc \\
=& (\one - \hat \Vc \hat \Vc^\dagger) \hat \Sc^{{1/2}^\dagger} \hat \Ac \hat \Sc^{1/2} \hat \Vc \\
=& (\hat \Sc^{{1/2}^\dagger} - \hat \Vc \hat \Vc^{\prime^\dagger}) \hat \Ac \hat \Vc^\prime.
\end{split}
\end{equation*}

Then, the pullback $\frac{d}{dt} \hat \Vc^\prime = \hat \Sc^{1/2} \frac{d}{dt} \hat \Vc$ is given by
\begin{equation*}
\begin{split}
\frac{d}{dt} \hat \Vc^\prime =& \hat \Sc^{1/2} (\hat \Sc^{{1/2}^\dagger} - \hat \Vc \hat \Vc^{\prime^\dagger}) \hat \Ac \hat \Vc^\prime \\
=& (\one - \hat \Vc^\prime \hat \Vc^{\prime^\dagger}) \hat \Ac \hat \Vc^\prime,
\end{split}
\end{equation*}
which is the Oja flow for a real-valued matrix $\hat \Ac$.
\qed
\end{pf}
Therefore, given a quantum system whose generator $\Lc$
exhibits a spectral gap between eigenvalues $\lambda_{r^2}$ and $\lambda_{r^2+1}$, we can apply
Oja's flow, in order to perform principal component analysis and extract the $r$ slowest degrees of freedom.
Then we obtain the following linear generator,
\begin{equation}
\hat \Lc_{\Vc_{\infty}} = \hat \Vc_\infty^\dagger \hat \Lc \hat \Vc_\infty \in \mathbb{C}^{r^2 \times r^2}.
\end{equation}
Exploiting this reduced model, the expectation values of observables of interest $O$ can be obtained by
\begin{equation}
\mathbb{E}_{\rho(t)}[O] = \bbra{O_{\Vc_\infty}} e^{\hat \Lc_{\Vc_\infty} t} \kket{\rho_{0_{\Vc_\infty}}},
\end{equation}
where $\rho_{0_{\Vc_\infty}}$ is the initial state of the reduced system and $\bbra{O_{\Vc_\infty}} = {\rm vec}(O)^\dagger \Vc_\infty$.

{\em Remark}: If we have an initial state $\rho_0$, such that
$\Pc \rho_0 = \rho_0$, where $\Pc = \Vc_\infty \Vc_\infty^\dagger$,
then
the simulation results are exact (to numerical precision).
Namely, since $\Vc_\infty$ is a steady-state of the Oja flow,
it is easy to verify that
$\|e^{\Lc t} \Pc - \Vc_\infty e^{\Lc_{\Vc_\infty}} \Vc_\infty^\dagger\|_\mathrm{sop} = 0$,
where $\|\cdot\|_\mathrm{sop}$ is any superoperator norm.
If $\Pc \rho_0 \neq \rho_0$, the error due to the projection decays
exponentially and is upper bounded by
$C e^{{\rm Re}(\lambda_{r^2+1}) t}$ for some $C > 0$.
The exact simulation property is the main advantage of
our method over AE. In order to achieve an exact simulation by
means of AE, one would need to compute a high order
perturbative series, which requires taking the inverse
of large $(n^2 - r^2) \times (n^2 - r^2)$ dense matrices,
even if $\hat \Lc$ is sparse.
Our method can be implemented using
in-place matrix operations only, limiting the total
memory requirement to that of:
one $n^2 \times n^2$ sparse matrix for $\hat \Lc$,
two $n^2 \times r^2$ dense matrices for $\hat \Vc$
and its derivative,
and one $r^2 \times r^2$ dense matrix for intermediate results.
Finally, the total number of floating-point operations required
per iteration of the integrator is of order
$\Oc(N^{\hat \Lc}_\mathrm{nz} r^2 + (n r^2)^2 + (n r)^2)$, where
$N^{\hat \Lc}_\mathrm{nz}$ is the amount of nonzero entries in $\hat \Lc$.

\section{Reduction of time-dependent generators}\label{sec:timedep}

Up to this point, we have only considered systems without an input.
Quantum control systems are typically bilinear control systems,
and the inputs are modeled by the inclusion of time-dependent
terms in the Hamiltonian.
In the simplest case, a controlled quantum system
has a generator of the form $\Lc(t) = \Lc_0 + u(t) \Lc_c$,
where $u(t)$ is the control input.
Since bilinearity is preserved in the reduced model
$\dot \rho = \Lc_\Vc(t) \rho$, one could take a static approach to model reduction
by computing $\hat \Vc$ for $\Lc_0$ and then
use $\Lc_\Vc(t) = \Lc_{0, \Vc} + u(t) \Lc_{c, \Vc}$.

However, this approach does not take into account the
time dependence of the eigenvalues $\lambda_i(t)$ of $\Lc(t)$,
nor that of the associated eigenspaces, which $\Vc$ is intended
to capture.
Motivated by deriving better reduced models for such controlled systems, we next discuss an extension of our
approach to the model reduction of time-dependent systems.

Following a scheme similar to that given in \cite{tranningerDetectabilityConditionsState2022} for the
reduced QR decomposition, and assuming that
$\Delta_{r^2} = \mathrm{Re}(\lambda_{r^2}(t)) - \mathrm{Re}(\lambda_{r^2+1}(t)) \geq c > 0,\ \forall t > t_0$,
we start by discretizing the time variable of the Lindbladian,
$t \in \{t_i\}_{i=0}^{f}$, and
approximate $\hat \Lc(t)$ with piecewise constant reduction matrix
$\hat \Vc_i$ in each interval $[t_i, t_{i+1})$ for $0 \leq i < f$.
Namely, let
\begin{equation}
\Lc_i = \Lc(t_i),
\end{equation}
then each interval $[t_i, t_{i+1}])$,
we update the reduced model performing Oja's flow for $A = \hat \Lc_i$,
and setting the initial value of the corresponding reduction matrix $\hat \Vc_i$ equal to the previous steady point
$\hat \Vc_{i-1, \infty}$, i.e.
\begin{equation}
\begin{split}
\frac{d}{dt} \hat \Vc_i =& (\one - \hat \Vc_i \hat \Vc_i^\dagger) \hat \Lc_i \hat \Vc_i^\dagger, \\
\hat \Vc_i(0) =& \hat \Vc_{i-1, \infty}.
\end{split}
\end{equation}
Then, each of the reduced models are given by
\begin{equation}
\hat \Lc_{\Vc_{i,\infty}} = \hat \Vc_{i, \infty}^\dagger \hat \Lc \hat \Vc_{i, \infty}.
\end{equation}

{\em Remark:} The time-discretization of the Lindbladian
does not need to match that of the state, i.e. the change of model
$\Lc_i$ does not need to be performed at every step of the simulation, but could be limited to a subset of such steps.
The same reduction matrix $\hat \Vc_{i, \infty}$ may be kept, for instance,
until the input signal $u(t)$ crosses certain threshold levels.

We also need to preserve the continuity of the state along the different changes of model.
To this end, we employ the following transition matrix $\hat \Wc_i$, which maps a state in the $i-1$-th model
to its equivalent representation in the $i$-th model.
\begin{equation}
\hat \Wc_i = \hat \Vc_{i, \infty}^\dagger \hat \Vc_{i-1, \infty}.
\end{equation}
However, depending on how the eigenvectors of $\Lc(t)$ evolve,
the transition matrices $\hat \Wc_i$ may not be trace preserving.
In fact, the following proposition tells us that updating the matrix reduction
matrix $\hat \Vc_{i, \infty}$ is useful only when the transition matrix $\hat \Wc_i$ is non-unitary.
Otherwise, the subspace that $\hat \Vc_{i+1, \infty}$ projects onto is identical to that of $\hat \Vc_{i, \infty}$.
\begin{prop}
The following three statements are equivalent:
\begin{enumerate}
\item The transition matrix $\hat \Wc_i$ is unitary.
\item The singular angles\footnote{The singular angles between $X, Y \in {\rm St}(\mathbb{C}^n, r)$ are given by arccosine of the singular values of $X^\dagger Y$. They are the angles between the subspaces spanned by the columns of $X$ and $Y$.} between $\hat \Vc_{i, \infty}$ and $\hat \Vc_{i+1, \infty}$ are all equal to zero.
\item The projectors $\hat \Vc_{i, \infty} \hat \Vc_{i, \infty}^\dagger$ and $\hat \Vc_{i+1, \infty} \hat \Vc_{i+1, \infty}^\dagger$ are equal to each other.
\end{enumerate}
\end{prop}
\begin{pf}
$(2) \Rightarrow (1)$ is automatic from the fact that $\hat \Vc_{i, \infty}^\dagger \hat \Vc_{i+1, \infty} = \hat \Wc_i$ would have all of its
singular values equal to one.
$(1) \Rightarrow (3)$ comes from the following calculation,
\begin{equation*}
\begin{split}
\hat \Wc_i^\dagger \hat \Wc_i = \hat \Vc_{i+1, \infty}^\dagger \hat \Vc_{i, \infty} \hat \Vc_{i, \infty}^\dagger \hat \Vc_{i+1, \infty} =& \one \Rightarrow \\
\hat \Vc_{i+1, \infty} \hat \Vc_{i+1, \infty}^\dagger \hat \Vc_{i, \infty} \hat \Vc_{i, \infty}^\dagger \hat \Vc_{i+1, \infty} =& \hat \Vc_{i+1, \infty} \Rightarrow \\
\hat \Vc_{i+1, \infty} \hat \Vc_{i+1, \infty}^\dagger \hat \Vc_{i, \infty} \hat \Vc_{i, \infty}^\dagger \hat \Vc_{i+1, \infty} \hat \Vc_{i+1, \infty}^\dagger =& \hat \Vc_{i+1, \infty} \hat \Vc_{i+1, \infty}^\dagger.
\end{split}
\end{equation*}
Then, since both $\hat \Vc_{i, \infty} \hat \Vc_{i, \infty}^\dagger$ and $\hat \Vc_{i+1, \infty} \hat \Vc_{i+1, \infty}^\dagger$ are projectors of the same rank, the previous
equation implies that they are equal.
$(3) \Rightarrow (2)$ is automatic as well. Since the subspaces are equal, the singular angles between them must be zero.
\qed
\end{pf}
Therefore, we also need to renormalize the state before continuing with the simulation.
Then, the reduced dynamics of the system in each interval $[t_{i},t_{i+1})$ are given by
\begin{equation}
\begin{split}
\frac{d}{dt} \kket{\rho_{\Vc_{i, \infty}}} =& \hat \Lc_{\Vc_{i, \infty}} \kket{\rho_{i_{\Vc_{i, \infty}}}}, \\
\kket{\rho_{\Vc_{i, \infty}}(t_i^+)} =& \frac{\hat \Wc_i \kket{\rho_{\Vc_{i-1, \infty}}(t_i^-)}}{\sqrt{r} \bra{\omega_r} \hat \Wc_i \kket{\rho_{\Vc_{i-1, \infty}}(t_i^-)}},
\end{split}
\end{equation}
where $\tr({\rm vec}^{-1}(\cdot)) = r \bra{\omega_r} \cdot$ and $\bra{\omega_r} = \frac{1}{\sqrt{r}} \sum_i \bra{i} \bra{i} = \frac{1}{\sqrt{r}} \bbra{\one}$.
Furthermore, if the time discretization of the Lindbladian coincides with
that of the state, then the evolution of the reduced model is given by
\begin{equation}
\label{eq:timedepsim}
\begin{split}
\kket{\rho_{\Vc_{i+1, \infty}}&(t_{i+1})} = \\
&\frac{\hat \Wc_{i+1}
e^{\hat \Lc_{\Vc_{i, \infty}} (t_{i+1} - t_i)}
\kket{\rho_{\Vc_{i, \infty}}(t_i)}}{
\sqrt{r} \bra{\omega_r} \hat \Wc_{i+1}
e^{\hat \Lc_{\Vc_{i, \infty}} (t_{i+1} - t_i)}
\kket{\rho_{\Vc_{i, \infty}}(t_i)}
}.
\end{split}
\end{equation}

{\em Remark}: We can study the error due to the difference between
$\Lc_i$ and $\Lc(t)|_{t \in [t_i, t_{i+1})}$ and due to the change
of model at time $t_{i+1}$. Let the step size of the simulation be
$\Delta t$, and let also $\Delta \Lc_i(t) = \Lc(t) - \Lc_i$.
In the following we assume $\Delta t$ to be small and
$\Lc(t)$ to be (approximately) constant within each $\Delta t$ interval.
These assumptions allow for the computation of the error
between simulations of the full and the reduced models, avoiding the need for introducing the time-ordered exponential integral
$\Tc e^{\int_{t_0}^{t_f} \Lc(\tau) d\tau},$ and simplifying the computation. The final assumption is that $\Pc_i \rho = \rho$ for
the particular step in which we compute the error, where
$\Pc_i = \Vc_{i, \infty} \Vc_{i, \infty}^\dagger$.
For a simulation step, we have that the error is upper bounded as
\begin{equation}
\label{eq:errorii}
\begin{split}
\|e^{\Lc(t) \Delta t} \Pc_i -&
\Vc_{i, \infty} e^{\Lc_{\Vc_{i, \infty}}(t) \Delta t} \Vc_{i, \infty}^\dagger\|_\mathrm{sop} \\
&\leq \Delta t \|\Qc_i \Delta \Lc_i(t) \Pc_i\|_\mathrm{sop}
+ \Oc(\Delta t^2).
\end{split}
\end{equation}
where $\Qc_i = \one - \Pc_i$.
The bound was computed by Taylor expansion of the exponential,
and the fact that $\Vc_{i, \infty}$ is a steady-state of the Oja flow.
We observe that the error \eqref{eq:errorii} is dominated
by $\Delta t$ and the leakage out of the support of $\Pc_i$
caused by $\Delta \Lc_i(t)$. If the input signal 
is continuous, then the leakage can be made arbitrarily
small by reducing $t_{i+1} - t_i$.
Furthermore, if $u(t)|_{t \to \infty} \to 0$ (as it is the case for example in stabilization problems)
then we know that given $\varepsilon$ there is a $t_i$
such that $\|\Delta \Lc_i(t)\| < \varepsilon,\, \forall t > t_i$.\\
Regarding the choice of the model update intervals, reducing $t_{i+1} - t_i$ not only increases
the computational demands of the simulation, but it also increases
the amount of times that we incur in an additional error
given by
\begin{equation}
\label{eq:errorij}
\|\Pc_i - \Vc_{i+1, \infty} \Wc_{i+1} \Vc_{i, \infty}^\dagger\|_\mathrm{sop} =
\|\Qc_{i+1} \Pc_{i}\|_\mathrm{sop}.
\end{equation}
However, faster switching implies a better ability to track the actual invariant subspaces and less error due to ``leakage'' outside the subspace.
For
continuous inputs, $\Qc_{i+1} \Pc_i \to 0$ as $t_{i+1} - t_i \to 0$
as long as $\Delta_{r^2} > 0,\, \forall t > t_0$.
Finally, if $\Lc(t)|_{t \to \infty}$ exists and admits a unique
stationary state
 then it is easy to see that the reduced simulation
 with $\Vc_\infty|_{u(t) = 0}$ is asymptotically exact.

\section{Ensuring CPTP dynamics}
\label{sec:lindblad_conds}

In this section we shall use the notation $\hat \Sc^{(n)} \in {\rm U}(n^2)$
to denote the ${\rm SWAP}$ gate, in order to make the dimension of the system it acts on explicit.
As such, $\hat \Sc^{(n)}$ acts on the original system, whereas $\hat \Sc^{(r)}$ acts on the
reduced system.
Let us also define the maximally entangled state $\ket{\omega} = \frac{1}{\sqrt{n}} \sum_i \ket{i}\ket{i} = \frac{1}{\sqrt{n}} \kket{\one}$
and the projectors $\Pi_\omega = \ketbra{\omega}{\omega}$ and $\Pi^\perp_\omega = \one - \Pi_\omega$.
Additionally, let us define the {\em $\Gamma$-involution} of $\hat \Xc = \sum_{ijkl} X_{ijkl} \ketbra{ij}{kl}$ by \[\hat \Xc^\Gamma = \sum_{ijkl} X_{ijkl} \ketbra{i}{k} \otimes \ketbra{j}{l}.\]

Then we can translate the three conditions for conditional CP that we recalled in Section \ref{sec:model} from superoperator to matrix form:
\begin{enumerate}
\item Hermicity \cite{zyczkowski2004duality}: $\hat \Sc^{(n)} \hat \Lc \hat \Sc^{(n)} = \hat \Lc^*$.
\item Trace preservation: $\hat \Lc^\dagger \ket{\omega} = 0$.
\item Conditional complete positivity \cite{wolfAssessingNonMarkovianDynamics2008}: $\Pi_\omega^\perp \hat \Lc^\Gamma \Pi_\omega^\perp \geq 0$.
\end{enumerate}

With these conditions in hand, we can now study what properties ${\hat \Vc}$ needs to satisfy so that
$T(\cdot) = {\hat \Vc}^\dagger \cdot {\hat \Vc}$ preserves them. The following proposition gives sufficient
conditions for Hermicity preservation.

\begin{prop}
\label{thm:hermicity}
Given $\hat \Lc$ such that $\hat \Sc^{(n)} \hat \Lc \hat \Sc^{(n)} = \hat \Lc^*$. This property is preserved by $T(\hat \Lc) = {\hat \Vc}^\dagger \hat \Lc {\hat \Vc}$ if
\begin{equation}
\hat \Sc^{(n)} \hat \Vc \hat \Sc^{(r)} = {\hat \Vc}^*.
\end{equation}

\end{prop}
\begin{pf}

We need that
\begin{equation*}
\hat \Sc^{(r)} ({\hat \Vc}^\dagger \hat \Lc {\hat \Vc}) \hat \Sc^{(r)} = ({\hat \Vc}^\dagger \hat \Lc {\hat \Vc})^*.
\end{equation*}

By using the facts that $\hat \Sc^2 = \one$ and $\hat \Sc^\dagger = \hat \Sc$, and the assumption that $\hat \Sc^{(n)} \hat \Lc \hat \Sc^{(n)} = \hat \Lc^*$, we have
\begin{equation*}
\begin{split}
\hat \Sc^{(r)} ({\hat \Vc}^\dagger \hat \Lc {\hat \Vc}) \hat \Sc^{(r)} =& (\hat \Sc^{(r)} {\hat \Vc}^\dagger \hat \Sc^{(n)}) (\hat \Sc^{(n)} \hat \Lc \hat \Sc^{(n)}) (\hat \Sc^{(n)} \hat \Vc \hat \Sc^{(r)}) \\
=& (\hat \Sc^{(n)} \hat \Vc \hat \Sc^{(r)})^\dagger \hat \Lc^* (\hat \Sc^{(n)} \hat \Vc \hat \Sc^{(r)}),
\end{split}
\end{equation*}
which is equal to $({\hat \Vc}^\dagger \hat \Lc {\hat \Vc})^*$ whenever $\hat \Sc^{(n)} \hat \Vc \hat \Sc^{(r)} = \hat \Vc^*$.
\qed

\end{pf}

The following proposition proves that such a Hermicity-preserving $\hat \Vc$ always exists and can be analytically computed, given
any steady state of the Oja flow.

\begin{prop}
\label{thm:symunitary}
Given a steady state $\hat \Vc$ of Oja's flow of the Lindbladian matrix $\hat \Lc$,
the operator $\hat \Wc = \sqrt{{\hat \Vc}^\dagger \hat \Sc^{(n)} {\hat \Vc}^* \hat \Sc^{(r)}}$ is unitary and $\hat \Vc \hat \Wc$ preserves Hermicity.
\end{prop}
\begin{pf}
$\hat \Lc$ satisfies $\hat \Sc^{(n)} \hat \Lc \hat \Sc^{(n)} = \hat \Lc^*$, then applying
the properties $\hat \Sc^2 = \one$ and $\hat \Sc = \hat \Sc^* = \hat \Sc^\top$, and taking the complex conjugate,
we have that
\begin{equation*}
\begin{split}
&(\one - {\hat \Vc} {\hat \Vc}^\dagger) \hat \Lc V = 0 \Rightarrow \\
&(\one - (\hat \Sc^{(n)} \hat \Vc \hat \Sc^{(r)})^* (\hat \Sc^{(n)} \hat \Vc \hat \Sc^{(r)})^\top) \hat \Lc (\hat \Sc^{(n)} \hat \Vc \hat \Sc^{(r)})^* = 0,
\end{split}
\end{equation*}
i.e. if ${\hat \Vc}$ is a steady state, then also $\hat \Sc^{(n)} {\hat \Vc}^* \hat \Sc^{(r)}$ is one. Furthermore, by linearity
$\frac{d}{dt} (\hat \Sc^{(n)} {\hat \Vc}(t)^* \hat \Sc^{(r)}) = \hat \Sc^{(n)} (\frac{d}{dt} {\hat \Vc}(t))^* \hat \Sc^{(r)}$. Therefore, if ${\hat \Vc}$ is stable, then
also $\hat \Sc^{(n)} {\hat \Vc}^* \hat \Sc^{(r)}$ is stable. Then, we conclude that $\hat \Sc^{(n)} {\hat \Vc}^* \hat \Sc^{(r)} \in \mathfrak{V}$ \cite{tsuzukiGlobalConvergenceOjas2025} and
${\hat \Vc} {\hat \Vc}^\dagger = (\hat \Sc^{(n)} {\hat \Vc} \hat \Sc^{(r)})^* (\hat \Sc^{(n)} {\hat \Vc} \hat \Sc^{(r)})^\top$. Finally, multiplying by ${\hat \Vc}^\dagger$ from the left and by ${\hat \Vc}$ from
the right, we have $\one = {\hat \Vc}^\dagger (\hat \Sc^{(n)} {\hat \Vc} \hat \Sc^{(r)})^* (\hat \Sc^{(n)} {\hat \Vc} \hat \Sc^{(r)})^\top {\hat \Vc} = \hat \Wc^2 \hat \Wc^{2^\dagger}$.
\qed
\end{pf}

The following proposition, in turn, shows that trace-preservation can also be guaranteed, given any Hermicity preserving $\hat \Vc \hat \Wc$.

\begin{prop}
Let $\hat \Vc$ be a steady state of Oja's flow and $\{\ket{\phi_i}\}_{i=1}^{r^2}$ an orthonormal basis of $\mathbb{C}^{r^2}$ with
$\ket{\phi_1} = \ket{\omega_r} = \frac{1}{\sqrt{r}} \kket{\one}$. Let $\{\ket{\psi_i}\}_{i=1}^r$
be another orthonormal basis of $\mathbb{C}^{r^2}$ such that $\ket{\psi_1}$ is a right-eigenvector of
$\hat \Lc_{\Vc}^\dagger$ with eigenvalue equal to zero, i.e. $\hat \Vc^\dagger \hat \Lc^\dagger \hat \Vc^\dagger \ket{\psi_1} = 0$.
Then $\hat \Wc = \sum_i \ketbra{\psi_i}{\phi_i}$ is unitary and $\hat \Lc_{\Vc \Wc}$ is trace preserving.
Furthermore, let $\hat \Vc^\prime = \hat \Vc \hat \Wc \sqrt{(\hat \Vc \hat \Wc)^\dagger (\hat \Sc^{(n)} \hat \Vc \hat \Wc \hat \Sc^{(r)})^*}$, then $\hat \Lc_{\Vc^\prime}$
is both Hermitian and trace preserving.
\end{prop}
\begin{pf}
It is easy to see that $\hat \Wc$ is unitary, due to the orthonormality of the two bases, then
\begin{equation*}
\begin{split}
\hat \Wc^\dagger \hat \Wc =& \sum_{ij} \ketbra{\phi_i}{\psi_i} \ketbra{\psi_j}{\phi_j}
= \sum_{ij} \delta_{ij} \ketbra{\phi_i}{\phi_j} = \one.
\end{split}
\end{equation*}
It is also easy to see that $\hat \Lc_{\Vc \Wc}$ is trace preserving, since $\hat \Wc$ maps $\ket{\omega}$ to the
kernel of $\hat \Lc_{\Vc}^\dagger$,
\begin{equation*}
\begin{split}
\hat \Lc_{\Vc \Wc}^\dagger \ket{\omega} =& (\hat \Wc^\dagger \hat \Vc^\dagger \hat \Lc \hat \Vc \hat \Wc)^\dagger \ket{\omega} \\
=& \hat \Wc^\dagger (\hat \Vc^\dagger \hat \Lc \hat \Vc)^\dagger \hat \Wc \ket{\omega}
= \hat \Wc^\dagger \hat \Lc_{\Vc}^\dagger \ket{\psi_1} = 0
\end{split}
\end{equation*}
Finally, in order to prove the last point, we have to show that $\ket{\omega}$ is an eigenvector of the unitary
$\hat \Wc^\prime = \sqrt{(\hat \Vc \hat \Wc)^\dagger (\hat \Sc^{(n)} \hat \Vc \hat \Wc \hat \Sc^{(r)})^*}$.
From Proposition \ref{thm:symunitary} we know that
$\hat \Vc^\dagger (\hat \Sc^{(n)} \hat \Vc \hat \Sc^{(r)})^* = \hat \Vc^\dagger \hat \Sc^{(n)} \hat \Vc^* \hat \Sc^{(r)}$
is unitary. Therefore, multiplying by $\hat \Sc^{(r)}$ from the right (and taking the complex conjugate),
we have that $\hat \Uc = \hat \Vc^\top \hat \Sc^{(n)} \hat \Vc$ is also unitary.
Furthermore, $\hat \Vc \ket{\psi_1}$, being a non-defective right-eigenvector of $\hat \Lc^\dagger$ with real eigenvalue, satisfies
$\hat \Sc \hat \Vc \ket{\psi_1} = c_1 \hat \Vc^* \ket{\psi_1^*}$, where $|c_i| = 1$ is due to the non-uniqueness of $\hat \Vc$.
Applying these properties, we have
\begin{equation*}
\begin{split}
\hat \Wc^\prime =& (\hat \Wc^\dagger \hat \Uc^\dagger \hat \Wc^* \hat \Sc^{(r)})^{1/2} \\
=& (\ketbra{\phi_1}{\psi_1} \hat \Vc^\dagger \hat \Sc^{(n)} \hat \Vc^* \ketbra{\psi_1^*}{\phi_1^*} \hat \Sc^{(r)} + \\
&\ \sum_{ij>1} \ketbra{\phi_i}{\psi_i} \hat \Uc^\dagger \ketbra{\psi_j^*}{\phi_j^*} \hat \Sc^{(r)})^{1/2}
\end{split}
\end{equation*}
\begin{equation*}
\begin{split}
=& (c_1 \ketbra{\omega}{\psi_1} \hat \Vc^\dagger \hat \Vc \ketbra{\psi_1}{\omega} + \\
&\ \sum_{ij>1} \ketbra{\phi_i}{\psi^\prime_i} \hat \Uc^\dagger \ketbra{\psi_j^{\prime^*}}{\phi_j^*} \hat \Sc^{(r)})^{1/2} \\
=& (c_1 \ketbra{\omega}{\omega} +
\sum_{ij>1} \ketbra{\phi_i}{\psi^\prime_i} \hat \Uc^\dagger \ketbra{\psi_j^{\prime^*}}{\phi_j^*} \hat \Sc^{(r)})^{1/2},
\end{split}
\end{equation*}
where we have also used $\ket{\psi_1^*} = \ket{\psi_1} = \ket{\omega}$, $\hat \Sc^{(r)} \ket{\omega} = \ket{\omega}$
and $\hat \Vc^\dagger \hat \Vc = \one$. Furthermore, by unitarity, we know that the rest of the sum runs over $i,j > 1$,
allowing us to conclude that $\ket{\omega}$ is an eigenvector of $\hat \Wc^\prime$ with eigenvalue $\sqrt{c_1}$. Then, finally,
we have
\begin{equation*}
\begin{split}
\hat \Lc_{\hat \Vc^\prime}^\dagger \ket{\omega} =&
\hat \Wc^{\prime^\dagger} \hat \Wc^\dagger \hat \Vc^\dagger \hat \Lc^\dagger \hat \Vc \hat \Wc \hat \Wc^\prime \ket{\omega} \\
=& \sqrt{c_1} \hat \Wc^{\prime^\dagger} \hat \Wc^\dagger \hat \Vc^\dagger \hat \Lc^\dagger \hat \Vc \hat \Wc \ket{\omega} \\
=& \sqrt{c_1} \hat \Wc^{\prime^\dagger} \hat \Wc^\dagger \hat \Vc^\dagger \hat \Lc^\dagger \hat \Vc \ket{\psi_1}
= 0,
\end{split}
\end{equation*}
which concludes the proof.
\qed
\end{pf}

In order to find a condition for CCP preservation, we first prove the following
lemma, which allows us to transform the CCP statement to a different form which does not
involve the $\cdot^\Gamma$ involution.

\begin{lem}
\label{thm:ccptr}
The CCP condition $\Pi_\omega^\perp \hat \Lc^\Gamma \Pi_\omega^\perp \geq 0$ is equivalent to
\begin{equation}
\label{eq:ccptr}
\begin{split}
\tr \left(\hat \Lc \left(x^* - \frac{\tr(x^*)}{n} \one \right) \otimes \left( x - \frac{\tr(x)}{n} \one \right) \right) &\geq 0 
\\ 
& \forall x \in \mathbb{C}^{n \times n} . \nonumber
\end{split}
\end{equation}
\end{lem}
\begin{pf}
By definition of positive definiteness, we have that the condition is satisfied if and only if
$\bra{x} \Pi_\omega^\perp \hat \Lc^\Gamma \Pi_\omega^\perp \ket{x} \geq 0$ for all $\ket{x} \in \mathbb{C}^{n^2}$.
By expanding $\Pi_\omega^\perp = \one - \ketbra{\omega}{\omega}$,
obtain the following four terms from the left-hand side of the inequality,
\begin{equation*}
\begin{split}
\bra{x} \hat \Lc^\Gamma \ket{x} &= \tr(\hat \Lc x^* \otimes x), \\
-\bra{x} \hat \Lc^\Gamma \ketbra{\omega}{\omega} \ket{x} &= -\frac{\tr(x^*)}{n} \tr(\hat \Lc \one \otimes x), \\
-\bra{x} \ketbra{\omega}{\omega} \hat \Lc^\Gamma \ket{x} &= -\frac{\tr(x)}{n} \tr(\hat \Lc x^* \otimes \one), \\
\bra{x} \ketbra{\omega}{\omega} \hat \Lc^\Gamma \ketbra{\omega}{\omega} \ket{x} &= \frac{|\tr(x)|^2}{n^2} \tr(\hat \Lc), \\
\end{split}
\end{equation*}
where we take $\ket{x} = {\rm vec}(x^\dagger)$. Finally, adding these terms together, we obtain \eqref{eq:ccptr}.
\qed
\end{pf}

We can now prove the following proposition, which gives sufficient conditions for the preservation of CCP.

\begin{prop}
\label{thm:positivity}
Given $\hat \Lc$ such that $\Pi_\omega^\perp \hat \Lc^\Gamma \Pi_\omega^\perp \geq 0$. This property is preserved by
$T(\hat \Lc) = {\hat \Vc}^\dagger \hat \Lc {\hat \Vc}$ if
\begin{equation}
{\hat \Vc} = V^* \otimes V.
\end{equation}
\end{prop}
\begin{pf}

We want that
\(
\Pi_\omega^\perp ({\hat \Vc}^\dagger \hat \Lc {\hat \Vc})^\Gamma \Pi_\omega^\perp \geq 0,
\)
i.e. for any $\ket{\check x}$, we need that
\begin{equation*}
\bra{\check x} \Pi_{\check \omega}^\perp ({\hat \Vc}^\dagger \hat \Lc {\hat \Vc})^\Gamma \Pi_{\check \omega}^\perp \ket{\check x} \geq 0,
\end{equation*}
while we know that
\(
\bra{x} \Pi_\omega^\perp \hat \Lc^\Gamma \Pi_\omega^\perp \ket{x} \geq 0,
\)
for any $\ket{x}$. Using Lemma \ref{thm:ccptr} we have that
\begin{equation}
\label{eq:positive_condition}
\tr\left({\hat \Vc}^\dagger \hat \Lc {\hat \Vc} \left( \check x^* - \frac{\tr(\check x^*)}{r} \check \one \right) \otimes \left( \check x - \frac{\tr(\check x)}{r} \check \one \right)\right) \geq 0,
\end{equation}
where
\begin{equation}
\label{eq:positive_known}
\tr \left(\hat \Lc \left(x^* - \frac{\tr(x^*)}{n} \one \right) \otimes \left( x - \frac{\tr(x)}{n} \one \right) \right)
\geq 0.
\end{equation}
Now we substitute ${\hat \Vc} = V^* \otimes V$ into \eqref{eq:positive_condition} and cycle the product inside the trace,
\begin{equation}
\label{eq:positive_ineq}
\begin{split}
\tr \left(\hat \Lc
\left( V \check x V^\dagger - \frac{\tr(\check x)}{r} \check \Pi \right)^* \otimes
\left(V \check x V^\dagger - \frac{\tr(\check x)}{r} \check \Pi \right) \right) \geq 0,
\end{split}
\end{equation}
where $\check \Pi = V V^\dagger$.

We have to prove that the previous inequality is satisfied by every $\check x$.
In order to prove it, we substitute $x = V \check x V^\dagger - \frac{1}{r} \tr(\check x) \check \Pi$
into \eqref{eq:positive_known}. Doing so results in the same inequality given in \eqref{eq:positive_ineq}, concluding the proof.
\qed
\end{pf}

\section{Complete-positivity preserving Oja's flow}
\label{sec:posflow}

As mentioned in Section \ref{sec:oja},
Oja's flow can be derived as the Riemannian gradient ascent flow of the following
cost-like function (not an actual cost function, since it is complex-valued)
on the Stiefel manifold ${\rm St}(\mathbb{C}^{n^2}, r^2)$,
\begin{equation}
J(\hat \Vc) = \frac{1}{2} \tr(A \hat \Vc \hat \Vc^\dagger).
\end{equation}
In this section we intend to develop a ``modified'' version of Oja's
flow that preserves complete positivity, i.e., such that ${\hat \Vc} = V^* \otimes V$ is guaranteed.
This modified flow also has the advantage of greatly reducing the amount of memory necessary to store the matrix $\hat \Vc$,
as only $\check V$ needs to be computed.

Substituting $A = \hat \Lc$ and the condition for $\hat \Vc$ in the $J(\hat \Vc)$ we obtain
a new cost function,
\begin{equation}
\begin{split}
J(V) =& \frac{1}{2} \tr(\hat \Lc (V^* \otimes V) (V^* \otimes V)^\dagger) \\
=& \frac{1}{2} \sum_m |\tr(L_m V V^\dagger)|^2 - r \tr(L_m^\dagger L_m V V^\dagger).
\end{split}
\end{equation}

Even though both $V$ and $\hat \Lc$ are both complex-valued matrices and $\hat \Lc \neq \hat \Lc^\dagger$,
this new functional is an actual (real-valued) cost function, unlike the previous $J(\hat \Vc)$. Therefore,
the gradient ascent of $J(V)$ on ${\rm St}(\mathbb{C}^n, r)$ recovers the optimization problem interpretation,
\begin{equation}
V_\infty = \underset{V}{\mathrm{argmax}} J(V).
\end{equation}
i.e. the steady states $V_\infty$ maximize the trace of the projection $(\Pi^* \otimes \Pi) \hat \Lc (\Pi^* \otimes \Pi)$,
with orthogonal projector $\Pi = V V^\dagger$. As such, the optimal $V_\infty$ is that which provides the
best approximation of the principal components of $\hat \Lc$.

The Riemannian gradient of $J(V)$, using the canonical metric of the Stiefel manifold, is given by
\begin{equation}
\begin{split}
{\rm grad} J(V) =&
\sum_{k} (\one - V V^\dagger)
(- r L_k^\dagger L_k + \\
& \quad \tr(L_k^\dagger V V^\dagger) L_k
+ \tr(L_k V V^\dagger) L_k^\dagger) V.
\end{split}
\end{equation}
Let $\check L_{k, V} = - r L_k^\dagger L_k + \tr(L_k^\dagger V V^\dagger) L_k +
\tr(L_k V V^\dagger) L_k^\dagger$, then the CCP preserving flow is given by

\begin{equation}
\epsilon \dot{V} = \sum_k (\one - V V^\dagger) \check L_{k, V} V,
\end{equation}
where $\epsilon$ is a rate-controlling parameter. The convergence of this flow to
(local) maxima of $J(V)$ is guaranteed by Zoutendijk's theorem \cite{satoRiemannianOptimizationIts2021}.

Notice that, with respect to the method introduced in Section \ref{sec:ojaL},
the memory footprint is reduced by a square factor: where
storing $n^2 \times r^2$ matrices was needed before, now only $n \times r$ matrices are required.
The amount of floating-point operations per iteration is also reduced to
$\Oc(N^{\check{L}_{V}}_\mathrm{nz} (1 + r + r^2) + n (r + r^2))$,
where $N^{\check{L}_{V}}_\mathrm{nz}$ is the amount of nonzero entries in
$\sum_k \check L_{k, V}$.
On the other hand,
while it remains true that a CCP reduced model of the slow or the center manifolds
not always exists, by enlarging the dimension of the reduced model, there is hope
to find a good approximation of the original system, while still obtaining a significant
reduction in its dimensionality. The following subsection provides some insight as to why
this is the case.

\subsection{Reconnecting to adiabatic elimination and DFS}

In this section, we develop some geometric intuition on the behavior of the CP-Oja flow and connect, when possible, the reduced models it produces with the corresponding AE ones.
Let $\Mc_{\alpha}$ be the set of density matrices contained in the span of the eigenoperators
of $\Lc$ whose eigenvalues have real part greater than or equal to $\alpha$, i.e.
\begin{equation}
\Mc_{\alpha} = \mathfrak{D}(\Hc)\cap({\rm span}(\{{\rm eigops_{\lambda_i}} | {\rm Re}(\lambda_i) \geq \alpha\})).
\end{equation}
It is possible to prove that the center manifold for the dynamics is spanned by $\Mc_0$ \cite{wolfQuantumChannelsOperations2012}, which, for Lindbladian dynamics,
has  the following general structure
\cite{baumgartnerAnalysisQuantumSemigroups2008,oreshkovAdiabaticMarkovianDynamics2010}: there exists a Hilbert space decomposition
$
\Hc = \bigoplus_{j=1}^m \Hc_{S,j} \otimes \Hc_{F,j} \oplus \Hc_R
$
such that $\Mc_0$ has the form
\begin{equation*}
 \{U \Big(\bigoplus_{j=1}^m p_j \rho_{S,j} \otimes \tau_{F,j} \oplus 0_R\Big) U^\dagger| 
 \rho_{S,j} \in \mathfrak{D}(\Hc_{S,j}),p\in\mathscr{C}_m\},
\end{equation*} for some $U \in \mathrm{U}(n),$ $m\in\mathbb{N},$ and fixed 
$\tau_{F,j} \in \mathfrak{D}(\Hc_{F,j}),$ with 
$ {\rm rank}(\rho_{F,j}) = {\rm dim}(\Hc_{F,j}),$ $\mathscr{C}_m=\{p\in\mathbb{R}^m|\sum_jp_j=1,\,p_j\geq 0\}.$
Each of the $\Hc_{S,j}$ experiences unitary dynamics, whereas each $\Hc_{F,j}$ supports a single full rank fixed state $\rho_{F,j}$
and $\Hc_R$ is a decaying subspace. If some $\Hc_{F,j}$ has dimension $1$, then the corresponding $\Hc_{S,j}$ is in particular
a DFS \cite{lidarQuantumErrorCorrection2013}.
Clearly, $\Mc_0 \subseteq \mathfrak{D}(\Hc_0)$, and that the equality occurs only if the Hilbert space
can be decomposed into a single DFS and the decaying subspace, i.e. $\Hc = \Hc_S \oplus \Hc_R$.

While the Oja flow tries to find a projection onto $\Mc_{\alpha}$,
the CP-ensuring version does so by means of reductions of the maps onto {\em subspaces of $\Hc$}.
To see this explicitly, take the reduced Lindbladian $\Lc_\Vc$ and apply, without loss of generality,
$\Vc=V^*\otimes V$ from left the and $\Vc^\dagger$ from the right, and define $\Pc_0 = \Vc \Vc^\dagger$.
We then have $\Vc \Lc_\Vc \Vc^\dagger = \Pc_0 \Lc \Pc_0$, where $\Pc_0$ is a projection
of the dynamics onto the subspace $\Hc_V=V V^\dagger \Hc $ of $\Hc$.
In other words, the reduced Lindbladian $\Lc_\Vc = V^\dag \mathcal{L}(V \cdot V^\dag)V$ 
corresponds to a reduction onto the subspace $\textrm{range}(V V^\dag)$.
This means that the resulting CP model reduction is exactly the intended slow dynamics only if
the target slow manifold corresponds to a single a block on $\Hc$, up to change of basis $U$:
more precisely, if $\Mc_\alpha=\mathfrak{D}(\Hc_0)$ and the rank of ${V},$ $r,$ is equal to $\dim(\Hc_0)$.
While the last situation is hard to obtain in practical situations, the reduction obtained with
the CP-Oja flow remains {\em asymptotically exact} if $\Mc_{0}$ has support on a nontrivial
subspace $\Hc_0\subsetneq \Hc,$ and $r \geq {\rm dim}({\rm supp}(\Mc_0))$. In this case, since the dynamics
exponentially converges to the center manifold and thus in particular to its support,
the CP-reduced dynamics become indistinguishable from the adiabatic as well as the exact ones in the long time limit.

In general, by finding the optimal reduction onto a subspace with sum of the reduced spectrum as close to zero as possible, the CP-Oja flow essentially
aims to find an {\em approximate DFS of the system}, providing an alternative numerical way to find approximate codes to the existing ones \cite{wangNumericalMethodFinding2013,casanovaQuantumWallStates2026,ticozziFindingQuantumNoiseless2007,choiMethodFindQuantum2006}.
If a DFS exists, then the projector $V_\infty V_\infty^\dagger$ onto it satisfies
$J(V_\infty) = 0$, since, as already mentioned, in this case $V V^\dagger$ is also a projector onto a subset of $\Mc_0$,
 and therefore the trace in $J(V_\infty)$ would simply sum
the corresponding eigenvalues of $\Lc,$ which lie on the imaginary axis and are either zero or complex-conjugate pairs.
Furthermore, since $\Lc$ has no eigenvalues with positive real part, that is also the maximum of $J(\cdot)$ and a steady
state of the CP-Oja flow.
This can be useful in scenarios where we can write 
$\Lc = \Lc_0 + \varepsilon \Lc_\varepsilon$, such that $\Lc_0$ exhibits a DFS of known dimension
and $\varepsilon > 0$ is small. In this situation, the perturbed dynamics will still admit
a slowly decohering, albeit no longer noise free, subspace.  Then, the CP-Oja flow allows us to
numerically find the projectors $\Pc_0(\cdot) = V V^\dagger \cdot V V^\dagger$ and
$\Qc_0 = \one - \Pc_0$ that project on such a slow subspace.
Finally, we can adapt the error bound \eqref{eq:errorii}
to this case, yielding
$
\|e^{\Lc t} \Pc_0 - \Vc e^{\Lc_{\Vc} t} \Vc^\dagger\|_\mathrm{sop}
\leq \Delta t |\varepsilon| \|\Qc_0 \Lc_\varepsilon \Pc_0\|_\mathrm{sop} + \Oc(\Delta t^2),
$
where $\Vc = V \cdot V^\dagger$.

\section{Numerical test - Central spin model}
\label{sec:simulations}

{\em System of interest:} In order to test the effectiveness of our methods, we consider a central spin $1/2$ model with a dissipative spin bath of $N_b=4$ spins $1/2$.
The Hilbert space of the system is thus $\Hc = \Hc_s \otimes (\bigotimes_i \Hc_{b,i})$, where
$\Hc_s \simeq \mathbb{C}^2$ represents the {\em central spin} and $\Hc_{b,i} \simeq \mathbb{C}^2$ the bath ones. Define the following the spin operators of the central spin:
\begin{equation*}
J_x = \frac{1}{2} \left[\begin{smallmatrix}
0 & 1 \\
1 & 0
\end{smallmatrix}\right] \otimes \one_b,
J_y = \frac{1}{2} \left[\begin{smallmatrix}
0 & -\iu \\
\iu & 0
\end{smallmatrix}\right] \otimes \one_b,
J_z = \frac{1}{2} \left[\begin{smallmatrix}
1 & 0 \\
0 & -1
\end{smallmatrix}\right] \otimes \one_b,
\end{equation*}
whereas those of the bath are given by
\begin{equation*}
\begin{split}
J_{x,i} =& \frac{1}{2} \one_s \otimes \one_2^{\otimes (i - 1)} \otimes \left[\begin{smallmatrix}
0 & 1 \\
1 & 0
\end{smallmatrix}\right] \otimes \one_2^{\otimes (N_b - i)}, \\
J_{y,i} =& \frac{1}{2} \one_s \otimes \one_2^{\otimes (i - 1)} \otimes \left[\begin{smallmatrix}
0 & -\iu \\
\iu & 0
\end{smallmatrix}\right] \otimes \one_2^{\otimes (N_b - i)}, \\
J_{z,i} =& \frac{1}{2} \one_s \otimes \one_2^{\otimes (i - 1)} \otimes \left[\begin{smallmatrix}
1 & 0 \\
0 & -1
\end{smallmatrix}\right] \otimes \one_2^{\otimes (N_b - i)}, \\
J_{\pm,i} =& J_{x,i} \pm \iu J_{y,i}.
\end{split}
\end{equation*}
\noindent {\em Dynamics:} The central spin has the following local Hamiltonian,
\begin{equation*}
H_s = \omega_s J_z,
\end{equation*}
whereas the one of the spin bath is
\begin{equation*}
H_b = \omega_b \sum_i J_{z,i} + \lambda_b \sum_{i,j} J_{x,i} J_{x,j}.
\end{equation*}
Finally, the interaction Hamiltonian is given by
\begin{equation*}
H_{\rm int} = A_x J_x \sum_i J_{x,i} + A_z J_z \sum_i J_{z,i},
\end{equation*}
with the total Hamiltonian being
\(
\label{eq:hamiltonian}
H = H_s + H_b + H_{\rm int}.
\)
As for the noise operators, dissipation is considered on the bath spins only,
$
L_m = J_{+,m}.
$

The following parameters were chosen
$\omega_s = 1.01$,
$A_x = 0.12$, $A_z = 0.03$,
$\omega_b = 1.92$, $\lambda_b = 0.31$, which results in a
generator with right-most eigenvalues given by the first column of Table \ref{tab:central_eig}.
With this choice of parameters, the system exhibits a spectral gap after the fourth eigenvalue.
This is due to the fact that the coupling between the central spin and the bath spins is relatively
weak, compared to the dissipation rate of the bath spins. The presence of this gap means that we can
obtain a good approximation of the dynamics of the central spin by AE.
\begin{table*}[!htpb]
\centering
\caption{Dominant eigenvalues of the generator $\hat \Lc$ of the full model and its reductions $\hat \Lc_{V_\infty} = V^\dagger \hat \Lc V$.}
\label{tab:central_eig}
\begin{tabular}{c|c|c|c|c|c}
\hline
    & Full & CP-Oja ($r = 2$) & CP-Oja ($r = 6$) & CP-Oja ($r = 10$) & Oja's Flow ($r = 2$) \\
\hline
$\lambda_1$ & $0$                    & $0$              & $0$                    & $0$                    & $0$ \\
$\lambda_2$ & $-0.0036 - 1.0673 \iu$ & $0 - 1.0700 \iu$ & $-0.0016 - 1.0682 \iu$ & $-0.0030 - 1.0678 \iu$ & $-0.0036 - 1.0673 \iu$ \\
$\lambda_3$ & $-0.0036 + 1.0673 \iu$ & $0 + 1.0700 \iu$ & $-0.0016 + 1.0682 \iu$ & $-0.0030 + 1.0678 \iu$ & $-0.0036 + 1.0673 \iu$ \\
$\lambda_4$ & $-0.0072$              & $0$              & $-0.0028$              & $-0.0060$              & $-0.0072$ \\
$\lambda_5$ & $-0.5002 - 1.7018 \iu$ &  -               & $-0.4988 - 1.8816 \iu$ & $-0.4975 - 1.7072 \iu$ &  -  \\
$\lambda_6$ & $-0.5002 + 1.7018 \iu$ &  -               & $-0.4988 + 1.8816 \iu$ & $-0.4975 + 1.7072 \iu$ &  -  \\
\hline
\end{tabular}
\end{table*}

\noindent {\em Numerical experiments:} We applied both Oja's flow and our positivity preserving variation to this system, resulting in reduced models
with eigenvalues as given by the rest of the columns of Table \ref{tab:central_eig}. The original Oja flow
perfectly captures the eigenvalues of the slow subspace with a reduced dimension of $2$, essentially
modeling a qubit with a non-Markovian environment. The CP-Oja flow returns a unitarily
evolving reduced model. The approximation provided by the CP-Oja flow improves by increasing the
dimension of the reduced model: Indeed, the right-most eigenvalues of the reduced models approach those of the full model as the dimension $r$ is increased.

In order to simulate the system, we choose the initial state to be a product state $\rho_0 = \rho_s \otimes \rho_b$,
where
$
\rho_b = e^{-H_b \beta} / \tr(e^{-H_b \beta})
$
is a thermal state with inverse temperature $\beta = 0.1$, and $\rho_s = \ketbra{+}{+}$ the
eigenstate of $J_x$ with eigenvalue $+1/2$, i.e. $2 J_x \ket{+} = \ket{+}$.

Figure \ref{fig:central} shows a comparison of the dynamics predicted by the {\em full model} of the system of dimension $n = 32$,
and the reduced models of dimension $r = 2$ given by {\em Oja's flow} and the {\em CP-Oja flow}. The expected value of the spin
along the $X$ axis given by the full model and the non-positive reduced model are identical, whereas the one given
by the positive reduced model experiences no dissipation and oscillates slightly faster. This is consistent with the
eigenvalues given in Table \ref{tab:central_eig}. As for the expectation along the $Z$ axis, the full model and
the non-positive reduced model exhibit differing behaviors for time close to $t = 0$. This is due to the original
state having support outside of the slow subspace. The positive reduced model, instead, predicts no dynamics along the $Z$ axis. This is another indication of the tradeoff between accuracy of the simulation and dimension of the model under the strict requests of maintaining the CP character of the dynamics.

\begin{figure}[ht]
\centering
\subfloat[][Expectation on $X$]{
\resizebox{0.8\linewidth}{!}{\begin{tikzpicture}
    \begin{axis}[%
            width = 1.2\linewidth,
            height = 1.2\linewidth/1.75,
            xmin=0.0,
            xmax=50.0,
            xlabel={Time ($t$)},
            ymin=-1.001,
            ymax=1.001,
            ylabel={X ($\langle \sigma_x \rangle$)},
            axis background/.style={fill=white},
            xmajorgrids,
            ymajorgrids,
            legend style={legend cell align=left, align=left},
            legend pos=south west
            ]
            \addplot [color1, solid, line width=1pt]
            table[x=t, y=Xf]{figures/central.dat};
            \addlegendentry{Full model ($n = 32$)}

            \addplot [color2, loosely dashed, line width=2pt]
            table[x=t, y=Xp]{figures/central.dat};
            \addlegendentry{Positive Oja ($r = 2$)}

            \addplot [color3, dotted, line width=2pt]
            table[x=t, y=Xo]{figures/central.dat};
            \addlegendentry{Oja flow ($r = 2$)}
            
    \end{axis}
\end{tikzpicture}}
\label{fig:central_x}
}\
\subfloat[][Envelope of the expectation on $X$]{
\resizebox{0.8\linewidth}{!}{\begin{tikzpicture}
    \begin{axis}[%
            width = 1.2\linewidth,
            height = 1.2\linewidth/1.75,
            xmin=0.0,
            xmax=1500.0,
            xlabel={Time ($t$)},
            ymin=-0.01,
            ymax=1.01,
            ylabel={X ($\langle \sigma_x \rangle$)},
            axis background/.style={fill=white},
            xmajorgrids,
            ymajorgrids,
            legend style={legend cell align=left, align=left},
            legend pos=north east
            ]
            \addplot [color1, solid, line width=1pt]
            table[x=t, y=Xf]{figures/central_ss.dat};
            \addlegendentry{Full model ($n = 32$)}

            \addplot [color2, loosely dashed, line width=2pt]
            table[x=t, y=Xp]{figures/central_ss.dat};
            \addlegendentry{Positive Oja ($r = 2$)}

            \addplot [color3, dotted, line width=2pt]
            table[x=t, y=Xo]{figures/central_ss.dat};
            \addlegendentry{Oja flow ($r = 2$)}
            
    \end{axis}
\end{tikzpicture}}
\label{fig:central_xss}
}\
\subfloat[][Expectation on $Z$]{
\resizebox{0.8\linewidth}{!}{\begin{tikzpicture}
    \begin{axis}[%
            width = 1.2\linewidth,
            height = 1.2\linewidth/1.75,
            xmin=0.0,
            xmax=1500.0,
            xlabel={Time ($t$)},
            ymin=-0.01,
            ymax=0.85,
            ylabel={Z ($\langle \sigma_z \rangle$)},
            axis background/.style={fill=white},
            xmajorgrids,
            ymajorgrids,
            legend style={legend cell align=left, align=left},
            legend pos=south east
            ]
            \addplot [color1, solid, line width=1pt]
            table[x=t, y=Zf]{figures/central_ss.dat};
            \addlegendentry{Full model ($n = 32$)}

            \addplot [color2, loosely dashed, line width=2pt]
            table[x=t, y=Zp]{figures/central_ss.dat};
            \addlegendentry{Positive Oja ($r = 2$)}

            \addplot [color3, dotted, line width=2pt]
            table[x=t, y=Zo]{figures/central_ss.dat};
            \addlegendentry{Oja flow ($r = 2$)}
            
    \end{axis}
\end{tikzpicture}}
\label{fig:central_zss}
}
\caption{
Comparison between the full model (solid blue line) and the reduced models
obtained by the CP-Oja flow (red dashed line) and the original Oja flow (green dotted line).
}
\label{fig:central}
\end{figure}
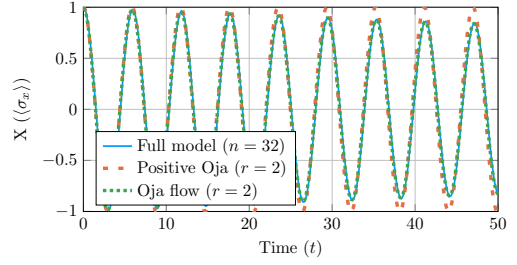
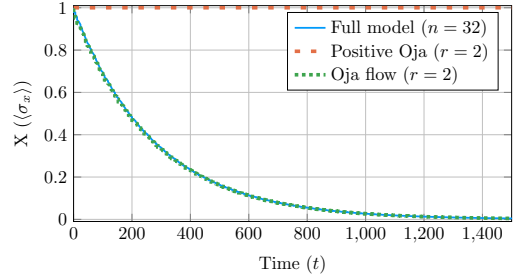
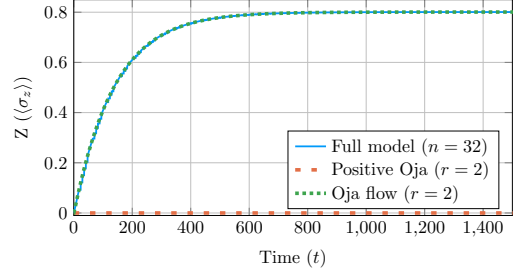

{\em Remark:} In this scenario, the CP-Oja flow effectively finds the best 2-dimensional subspace code for the dynamics: it corresponds to
$\Kc_0 = \mathcal{H}_s \otimes {\rm span}\{\ket{0000}_b\}$, i.e. the central spin degrees of freedom tensor the span of the steady state of the dissipative part of the dynamics (recall that the bath spins decay with $J_p$). Notice that this subspace, thanks to its factorized form, can be interpreted as the supporting of a general subsystem encoding the full reduced state of the central spin \cite{knillTheoryQuantumError2000}.

Figure \ref{fig:central_w} shows the same simulation but for positive reduced models with different dimensions $r$. 
For $r > 2$ the reduced models start being able to exhibit a dissipative behavior.
For $r = 10$, the CP-Oja flow converges to $V_\infty$ such that $V_\infty V_\infty^\dagger$
is a projector onto the direct sum of $\Kc_0$ and 
$\Kc_1 = \mathcal{H}_s \otimes {\rm span}\{\ket{0001}_b, \ket{0010}_b, \ket{0100}_b, \ket{1000}_b\}$, with $\Kc_0\oplus\Kc_1=\Hc_S\otimes \Kc,$ an invariant subspace of the dissipative part of the generator that is capable of encoding the reduced state on the central spin.

On the other hand, for $2 < r < 10$, the algorithm finds a $V_\infty$ with support within $\Kc_0 \oplus \Kc_1$,
but such that $V_\infty V_\infty^\dagger$ projects onto $\Kc_0 \oplus \tilde\Kc_1$,
where $\tilde\Kc_1$ is a non-unique subspace of $\Kc_1$ (different initializations
of the CP-Oja flow yield different $\tilde\Kc_1$).
In this case, the increase in dimension, e.g. from $r=2$ to $r = 6$,
allows the CP-Oja flow to better approximate the eigenvalues
of the generator, as shown in Table \ref{tab:central_eig}. However, it also
introduces an error in the expectation of the observable of interest $\langle\sigma_x\rangle$, due to the fact that the subspace $\Kc_0\otimes\tilde \Kc_1$ does not allow for a factorized form as $\Hc_s\otimes \Kc$, failing to encode the  intial reduced state of the center spin and yielding the poorer tracking performance displayed in Figure \ref{fig:central_w}.
We have verified that for  $r \in \{2, 10, 22, 30, 32\}$,
the method projects onto factorized invariant subspaces of the
dissipative part of the dynamics, avoiding such initialization-driven errors. These observations point to $r$ as a key parameter when using the CP-preserving algorithm.

\begin{figure}[ht]
\centering
\subfloat[][Expectation on $X$]{
\resizebox{0.8\linewidth}{!}{\begin{tikzpicture}
    \begin{axis}[%
            width = 1.2\linewidth,
            height = 1.2\linewidth/1.75,
            xmin=0.0,
            xmax=50.0,
            xlabel={Time ($t$)},
            ymin=-1.001,
            ymax=1.001,
            ylabel={X ($\langle \sigma_x \rangle$)},
            axis background/.style={fill=white},
            xmajorgrids,
            ymajorgrids,
            legend style={legend cell align=left, align=left},
            legend pos=south west
            ]
            \addplot [color1, solid, line width=1pt]
            table[x=t, y=Xf]{figures/central_w.dat};
            \addlegendentry{Full model ($n = 32$)}

            \addplot [color2, loosely dashed, line width=2pt]
            table[x=t, y=X1]{figures/central_w.dat};
            \addlegendentry{Positive Oja ($r = 2$)}

            \addplot [color3, dotted, line width=2pt]
            table[x=t, y=X3]{figures/central_w.dat};
            \addlegendentry{Positive Oja ($r = 6$)}
            
            \addplot [color4, dashdotdotted, line width=2pt]
            table[x=t, y=X5]{figures/central_w.dat};
            \addlegendentry{Positive Oja ($r = 10$)}
    \end{axis}
\end{tikzpicture}}
\label{fig:central_xw}
}
\caption{Comparison between the full model (solid blue line) and reduced models obtained
by the CP-Oja flow for different dimensions $r$. Notice that, as dimensions, neither $r = 6$ nor $r = 10$ are
compatible with a simple tracing out of some of the bath spins.}
\label{fig:central_w}
\end{figure}
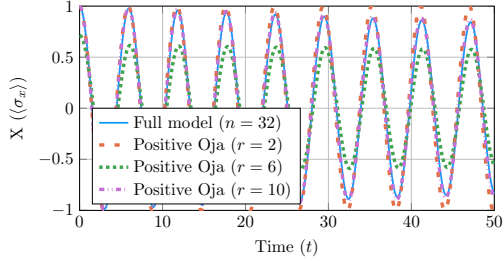

We now proceed to test the Oja flow on a time-varying system. To this end,
we modify Hamiltonian \eqref{eq:hamiltonian} by adding a Hamiltonian drive along
the $Y$-axis of the central spin with a Gaussian envelope, i.e.
\begin{equation}
\label{eq:hamiltonian_t}
H_d(t) = H + u(t) J_y
\end{equation}
where $u(t) = E_y e^{- (t - \tau_0)^2 / (2 \tau_y^2)}$,
and where $E_y = 1.28$, $\tau_y = 5$ and $\tau_0 = 25$.

Figure \ref{fig:central_t} shows the result of the simulation of a system
with Hamiltonian \eqref{eq:hamiltonian_t} and the same initial state as in
the previous simulations.
When the model reduction matrix is computed only at the start
of the simulation, the error of the
expectation of $Z$ becomes noticeable after $t = 30$.
On the other hand, if the model reduction matrix is re-computed
either at every step of the simulation, or only when the
input $u(t)$ crosses thresholds set at
$\{0.0, 0.3, 0.6, 0.9, 1.2\}$, the error with respect to the
simulation of the full system is negligible.

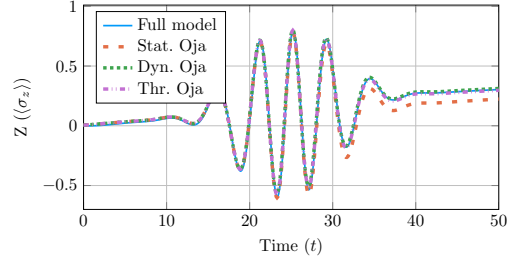
\begin{figure}[ht]
\centering
\subfloat[][Expectation on Z]{
\resizebox{0.8\linewidth}{!}{\begin{tikzpicture}
    \begin{axis}[%
            width = 1.2\linewidth,
            height = 1.2\linewidth/1.75,
            xmin=0.0,
            xmax=50.0,
            xlabel={Time ($t$)},
            ymin=-0.7,
            ymax=1.01,
            ylabel={Z ($\langle \sigma_z \rangle$)},
            axis background/.style={fill=white},
            xmajorgrids,
            ymajorgrids,
            legend style={legend cell align=left, align=left},
            legend pos=north west
            ]
            \addplot [color1, solid, line width=1pt]
            table[x=t, y=Zf]{figures/central_t.dat};
            \addlegendentry{Full model}

            \addplot [color2, loosely dashed, line width=2pt]
            table[x=t, y=Zp]{figures/central_t.dat};
            \addlegendentry{Stat. Oja}

            \addplot [color3, dotted, line width=2pt]
            table[x=t, y=Zo]{figures/central_t.dat};
            \addlegendentry{Dyn. Oja}

            \addplot [color4, dashdotdotted, line width=2pt]
            table[x=t, y=Zt]{figures/central_t.dat};
            \addlegendentry{Thr. Oja}
    \end{axis}
\end{tikzpicture}}
\label{fig:central_zt}
}
\caption{
Comparison between the full time-dependent model (solid blue line) and the reduced models
obtained by the Oja flow performed at $t=0$ only (red dashed line), the Oja flow performed throughout
the entire simulation (green dotted line), and the Oja flow performed
only when the input signal crossed thresholds set at multiples of $0.3$ (purple dashed and dotted line).
}
\label{fig:central_t}
\end{figure}

The autonomous system simulation took $11\mathrm{s}$
on a commercial laptop, while those of the reduced systems
took $0.4\mathrm{s}$ each. To the latter, one has to add
the pre-processing time dedicated to computing the reduced model,
which was $3\mathrm{s}$ for the Oja flow,
$0.05\mathrm{s}$ for the CP-Oja flow with $r = 2$ and
$0.1\mathrm{s}$ with $r = 10$.
The simulation of the controlled system took 7:03 minutes
for the full system, $36\mathrm{s} + 3\mathrm{s}$ for the
system with static $\Vc$, 3:48 minutes for system with
$\Vc_i$ updated at each step and $34\mathrm{s} + 15\mathrm{s}$ for
the case with $\Vc_i$ updated only when $u(t)$ crossed multiples
of $0.3$.

\section{Conclusions}

In this work, a new approach to finding and approximating the slow dynamics of quantum Markovian dynamics is developed, which bypasses the need for Jordan decomposition of the Markov dynamical generator by finding its optimal {\em slow operator or vector subspace} by integrating a dynamical equation - the Oja flow.
The reduced generator can then be used to obtain asymptotically exact time-traces of the observable of interests, or reduced state dynamics for large systems.
In general, with respect to the slow {\em operator} subspace reduction, the CP-preserving method we present needs a larger reduced model to accurately capture the key features of the dynamics - as expected by the constraints imposed on the optimization variable, which correspond to a reduction of the system underlying {\em vector} subspace.
Yet, we have shown in the numerical examples that it is possible
to obtain a good approximation of the dynamics by enlarging the dimension of the reduced model while still obtaining
a significant reduction of the dimension of the model. The approach has also potential applications to finding approximate DFS for a system, bypassing some numerical difficulties related to Jordan and algebraic decompositions.

In future work, we aim to apply the techniques proposed in this paper
to the design of numerically efficient strategies for quantum optimal control,
for instance, by substituting the dynamical constraints by their reduced version.
The CP-preserving flow may also have promising applications in the field
of quantum information protection, thanks to its connection to approximate DFS.
Recent work done in \cite{ohkiLowRankApproximationKalmanBucy2026}
may also be used to further refine our strategy for time-varying
systems.
We also aim to generalize this approach, defining a modified Oja flow that shifts the optimization from ${\hat \Vc} = V^* \otimes V$ to ${\hat \Vc} = \sum_i^k v_i V_i^* \otimes V_i$ for fixed $k$.
While it may not be easy to ensure a CP-preserving flow, it could lead to an approximation of the Oja flow that is
significantly less memory intensive by, essentially,
truncating $\hat \Vc$ to its $k$ principal components.
Furthermore, we plan to exploit the discrete-time version of Oja's flow to perform model reduction of CPTP maps. An algorithm for this purpose has been recently proposed  in \cite{tsuzukiConvergenceNPM2025}, which converges to the dominant eigensubspace associated with eigenvalues of the largest absolute value, effectively capturing the slow modes of stable linear systems in discrete time. A CP-preserving algorithm for discrete quantum systems could be derived leveraging these results.
Lastly, the proposed methods will be applied to experimentally-motivated, challenging systems and compared more thoroughly, in particular in terms of computational complexity, to the AE methods in perturbative scenarios.

\section*{Acknowledgments}

F. Ticozzi and M. Casanova were supported by the European Union through
NextGenerationEU, within the National Center for HPC, Big Data and
Quantum Computing under Projects CN00000013, CN 1, and Spoke
10. K. Ohki was supported by JSPS KAKENHI Grant Numbers JP21K12097, JP23K26126 and JP26K07550.

\appendix

\section{Derivation of the adiabatic generator}\label{sec:example}
 As a paradigmatic example, consider a system that can be decomposed into two weakly-coupled subsystems, where one
is strongly dissipative with a single steady state $\rho_{Q, \infty}$ and the other either decays much slower.
As such, let the Hilbert space be decomposed as $\Hc = \Hc_P \otimes \Hc_Q$, and
the system's dynamics are governed by the following Lindbladian,
$
\Lc = \Lc_Q + \varepsilon \Lc_P + \varepsilon \Lc_{\rm int},
$
where $\Lc_Q = \one \otimes \tilde{\Lc}_Q$ and $\Lc_P = \tilde{\Lc}_P \otimes \one$.

We now wish to perform a perturbative expansion on $\varepsilon$ in order to find the projector
$\Pc = \sum_{i=0} \varepsilon^i \Pc_i$ onto the slow subspace and the effective Lindbladian
$\Lc_{\rm eff} = \sum_{i=0} \varepsilon^i \Lc_i \Pc_0$, where $\Pc_0(\cdot) = \tr_Q(\cdot) \otimes \rho_{Q, \infty}$
and $\Lc_0 = 0$. Then, substituting the expansions into
$\Lc \Pc \rho = \Pc \Lc_{\rm eff} \rho$ and separating by orders of $\varepsilon$ we have
\begin{equation*}
\begin{split}
\Lc_Q \Pc_1 + \Lc_{\rm int} \Pc_0 + \Lc_{P} \Pc_0 =& \Pc_0 \Lc_1 \\
\Lc_Q \Pc_2 + \Lc_{\rm int} \Pc_1 + \Lc_{P} \Pc_1 =& \Pc_0 \Lc_2 + \Pc_1 \Lc_2.
\end{split}
\end{equation*}
It is easy to verify that $\Pc_0 \Lc_Q = 0$, then we have
$
\Lc_1 = \Pc_0 \Lc_{\rm int} \Pc_0 + \Pc_0 \Lc_P \Pc_0
= \frac{1}{\varepsilon} \Pc_0 \Lc \Pc_0.
$
In order to find $\Pc_i$ and $\Lc_{i+1}$ for $i \geq 1$, inversion of $\Lc_Q$ is necessary.
Explicit formulas for the case when $\Lc_P$ is purely Hamiltonian are given in \cite{rivaExplicitFormulasAdiabatic2024}.
Given the inverse, then we can compute the first order term of the projector $\Pc$ as
$
\Pc_1 = - \frac{1}{\varepsilon} \Lc_Q^{-1} \Lc \Pc_0,
$
and the second order term of the effective generator as
$
\Lc_2 = - \frac{1}{\varepsilon^2} \Pc_0 \Lc \Qc_0 (\Qc_0 \Lc_Q \Qc_0)^{-1} \Qc_0 \Lc \Pc_0,
$
where $\Qc_0 = \one - \Pc_0$. Adding these terms together we have that up to second order the effective generator has
the following form,
\[
\Lc_{\rm eff} = \Pc_0 \Lc \Pc_0 -
\Pc_0 \Lc \Qc_0 (\Qc_0 \Lc_Q \Qc_0)^{-1} \Qc_0 \Lc \Pc_0.
\]

\bibliography{refs}

\end{document}